\documentclass[cmp]{svjour_no_cmp}
\usepackage{latexsym,stmaryrd}
\usepackage{graphics}
\usepackage{hyperref, enumerate}
\usepackage{soul}
\usepackage{mathtools}
\usepackage{graphicx}
\usepackage{amsmath}
\usepackage{amssymb} %amsthm}
\usepackage{verbatim}
\usepackage[usenames,dvipsnames,svgnames,table]{xcolor}
\usepackage{cancel}
\usepackage{nicefrac}

\def\C{\mathbb C}

\def\Tr{\mathop {\rm Tr}}

\def\calm{\mathcal{M}}
 
\def\ep{\epsilon}
\def\cz{\circ}
\def\rz{R, \vartheta}
\def\be{\begin{equation}}
\def\beq{\begin{eqnarray}}
\def\beqs{\begin{eqnarray*}}
\def\ee{\end{equation}}
\def\eeq{\end{eqnarray}}
\def\eeqs{\end{eqnarray*}}
\def\eqref#1{(\ref{#1})}
\def\lra#1{\langle #1 \rangle}
\def\lrp#1{\left( #1 \right)}

\def\abs#1{\vert #1\vert}
\def\B{\mathcal{B}}
\def\A{\mathfrak{A}}

\def\cS{\mathcal{S}}
\def\cP{\mathcal{P}}

\def\even{{\rm even}}
\def\odd{{\rm odd}}

\def\nn{\nonumber\\}
\def\Z{\mathbb{Z}}
\def\T{\mathbb{T}}
\def\R{\mathbb{R}}
\def\I{\mathfrak{I}}

\def\qsp#1{\quad\text{#1}\quad}
\def\half{{\textstyle \frac{1}{2}}}
\def\cH{\mathcal{H}}

\def\rz{R, \Theta}
\def\hz{H, \Theta}

\def\tr{{\rm T}}
\def\:{{:}}
\def\vep{\varepsilon}
\def\vn{\varnothing}

\def\bm#1#2#3#4{\begin{pmatrix}#1&#2\\#3&#4\end{pmatrix}}

\renewcommand{\leq}{\leqslant}
\renewcommand{\geq}{\geqslant}
\newtheorem{thm}{Theorem}
\newtheorem{cor}[thm]{Corollary}
\newtheorem{lem}[thm]{Lemma}
\newtheorem{prop}[thm]{Proposition}
\newtheorem{defini}[thm]{Definition}
\newtheorem{Remark}[thm]{Remark}

\numberwithin{equation}{section}
\numberwithin{thm}{section}

\begin{document}
\title{Characterization of Reflection Positivity: 
Majoranas and Spins}
\author{Arthur Jaffe\inst{1} \and Bas Janssens\inst{2}% etc
}                    

\institute{Harvard University,
Cambridge, MA 02138, USA. \\Email: arthur\_jaffe@harvard.edu
 \and Universiteit Utrecht, 3584 CD Utrecht, The Netherlands. \\Email: b.janssens@uu.nl}
\maketitle

\begin{abstract}
We study linear functionals on a Clifford algebra (algebra of Majoranas)
equipped with a reflection automorphism. For Hamiltonians that are functions of Majoranas or of spins, we find necessary and sufficient conditions on the coupling constants for reflection positivity to hold. One can easily check these conditions in concrete models.  We illustrate this by discussing a number of spin systems with nearest-neighbor and long-range interactions.
\end{abstract}

\section{Introduction}  \label{Sect:Introduction}
We consider a finite-dimensional $\Z_{2}$-graded $*$-algebra 
$\A = \A^{\even}\oplus\A^{\odd}$.
The  algebra $\A$ is a graded (super) tensor product of two algebras $\A_{\pm}$, 
related by an anti-linear automorphism $\Theta \colon \A \rightarrow \A$, satisfying $\Theta(\A_{\mp})=\A_{\pm}$ and $\Theta^2 = I$.  
In this sense, $\A$ is the \emph{double} of $\A_{+}$. 
Such automorphisms often arise from geometric reflections on an underlying manifold, so
we refer to $\Theta$ as the \emph{reflection automorphism}.
The main results summarized in Theorem \ref{Theorem:Summary}  do not refer to an underlying geometry--while in the examples of \S \ref{sec:exspin} this becomes relevant. 

In this context, we are interested in \emph{even}
linear functionals $\omega \colon \A \rightarrow \C$ that are 
both \emph{reflection invariant} and \emph{reflection positive} with respect to the reflection $\Theta$.
A functional is called even if $\omega(\A^{\rm odd}) = 0$, and 
just like in the ungraded case, 
it is called
reflection invariant if
$\omega ( \Theta (A)) = \overline{\omega(A)}$
for all $A \in \A$.
The notion of reflection positivity has to be adapted to the $\Z_{2}$-grading;
we call $\omega$ reflection positive on $\A_{+}$
if
\begin{eqnarray}
0 &\leq& \omega(\Theta(A)A)\;, \qsp{\phantom{z}for} A \in \A_{+}^{\rm even}, \label{eq:hilb}\\
0 &\leq& \zeta \,\omega(\Theta(A)A)\;, \qsp{for} A \in \A_{+}^{\rm odd}\nonumber\,,
\end{eqnarray}
where $\zeta = \pm i$ is fixed once and for all.  
We introduce the \emph{twisted product}  
$\circ: \A_{-}\times \A_{+}\mapsto \A$ with
	\be\label{eq:Twisted-Product}
	A\circ B = \left\{\begin{matrix}
			AB\;,\phantom{\, \zeta}& \text{if $A$ or $B\in\A^{\even}$ \phantom{b}}\hfill\\
			\zeta\, AB\;,& \text{\,if both $A, B\in\A^{\odd}$}\hfill
			\end{matrix}\right.\;.
	\ee	
In terms of this twisted product, the reflection positivity equation \eqref{eq:hilb} becomes simply
	\be\label{eq:RPcirc}
		0\leq \omega(\Theta(A)\circ A) 
		\;, \qsp{\phantom{z}for} A \in \A_{+}\;.
	\ee

This definition of a reflection-positive form is natural  in the context of super algebras.  The completions of $\A_{+}^{\rm even}$ and $\A_{+}^{\rm odd}$ with respect to the form  
\eqref{eq:RPcirc} are then the orthogonal, even and odd parts of a \emph{super-Hilbert space} $\cH$, see Deligne and Morgan \cite{SuperSolutions}.  We elaborate on this  relation to super-Hilbert spaces in \S\ref{Sect:Reflection-InvariantFunctionals}.

We consider in detail  the case that $\omega=\omega_{H}$ is a \emph{Boltzmann functional}.  By this we mean that there is an element $H\in\A$ called the Hamiltonian,  such that 
	\[
	\omega_{H}(A)=\Tr(A\,e^{-H})\,,
	\]
where $\Tr$ is a tracial state on $\A$.  
If the \emph{partition sum} $Z_{H} := \Tr(e^{-H})$ is nonzero, 
define the {\em Gibbs functional} $\rho_{H}$ as the normalized Boltzmann functional,  
	\be\label{eq:introgibbs}
		\rho_{H}(A) = Z^{-1}_{H} \Tr(A\,e^{-H})\;.
	\ee

In statistical physics, $H \in \A$ is self-adjoint. In this case $Z_{H} >0$, and $\rho_{H}$ is a state, meaning that $\rho_{H}$ is positive and normalized.  Furthermore, 
it has the KMS property with respect to 
the automorphisms of $\A$ induced by $e^{itH}$. But in lattice approximations to fermionic quantum fields, the action plays the role of $H$ and often is not hermitian. In any case we do not assume that $H$ is hermitian.

Here we specialize to two types of algebras $\A$. 
In the first part of the paper, \S \ref{Sect:Introduction}--\S \ref{Sect:Which},
$\A$ will be an algebra of Majoranas, whereas in the second part
\S \ref{sec:RPSpinSyst}--\S \ref{sec:exspin}, $\A$ will generally be an algebra of spins.

An algebra of Majoranas is a $*$-algebra 
generated by self-adjoint operators $c_{i}$.
They are
labeled by indices $i$ running over a 
finite set $\Lambda$, and
satisfy the Clifford relations
\be\label{eq:Cliffrels}
{c_i}{c_j}+c_{j}c_{i}=2\delta_{ij}I\,, \quad i,j \in \Lambda\,.
\ee  
The  $\Z_{2}$ grading of $\A$ is defined as $+1$ on the even and $-1$
on the odd monomials in the $c_{i}$.
Even elements of $\A$ are often called {\em globally gauge invariant}.

The reflection automorphism $\Theta$ of the Majorana algebra $\A$ 
comes from a fixed point free reflection
$\vartheta \colon \Lambda \rightarrow \Lambda$.
If $\Lambda$ is the disjoint union of $\Lambda_{+}$
and $\Lambda_{-}$ with 
$\vartheta(\Lambda_{\pm}) = \Lambda_{\mp}$, then
the algebras $\A_{\pm}$ are generated by 
the Majoranas $c_{i}$ with $i\in \Lambda_{\pm}$.
In many applications, $\Lambda$ will be a finite lattice in Euclidean space,
and $\vartheta$ the reflection in a hyperplane which does not intersect $\Lambda$.

We give necessary and sufficient conditions such that the functionals $\omega_{H}$ and $\rho_{H}$ are reflection positive on $\A_{+}$.
Every Hamiltonian $H\in \A$ is defined by a coupling-constant matrix $J$ as
\be
H= -\sum
J_{i_{1}, \ldots, i_{k} \,;\, i'_{1}, \ldots, i'_{k'}}
\,
\Theta(c_{i_{1}}\cdots c_{i_{k}}) \circ
(c_{i'_{1}}\cdots c_{i'_{k'}})\,,
\ee
where $k$ and $k'$ range over $\mathbb{N}$, and 
$i_{1}, \ldots, i_{k}$ and $i'_{1}, \ldots, i'_{k'}$ 
range over $\Lambda_{+}$.
In fact one restricts the set over which one sums,  in order
to make the expansion unique, as explained in 
\S\ref{sec:bases}--\S\ref{Sect:TwistedProduct}. 
The conditions on reflection positivity are expressed in terms 
of the submatrix $J^{0}$ of $J$ for which $k \neq 0$ and $k' \neq 0$.  
If $\vartheta$ comes from a reflection in Euclidean space,  
$J^{0}$ describes the \emph{couplings across the reflection plane}.
We use this terminology even if a geometric interpretation is lacking.

The central result in this paper, which also holds with $\rho_{\beta H}$ replaced by 
$\omega_{\beta H}$, is the following:  
\begin{thm}\label{Theorem:Summary}
Let $H$ be reflection invariant and globally gauge invariant.  Then $\rho_{\beta H}$ is reflection positive for all $0<\beta$,  if and only if  $0\leq J^{0}$. 
\end{thm}

In the second part of the paper, we focus on spin algebras $\A^{\rm spin}$, generated 
by the Pauli matrices $\sigma^{1}_{j}$, $\sigma^{2}_{j}$, $\sigma^{3}_{j}$ 
associated to each lattice site $j \in \Lambda$.  
In \S \ref{sec:RPSpinSyst}
we study Hamiltonians of the form 
\be\label{eq:introspinham}
H^{\rm spin}= -
\sum
J{}^{a_1, \ldots, a_k}_{i_{1}, \ldots, i_{k}}
\,
\sigma_{i_{1}}^{a_{1}}\cdots \sigma_{i_{k}}^{a_{k}}\;.
%(\sigma_{i'_{1}}^{a'_{1}}\cdots \sigma_{i'_{k'}}^{a'_{k'}})\,,
%\sum_{\I\I'} J_{\I\I'}\,W_{\I\I'}\,,
\ee
By expressing the spins  $\sigma^{a}_{j}$ as even polynomials in 
the Majoranas, we translate Theorem \ref{Theorem:Summary} to the 
spin context. This yields necessary and sufficient conditions 
on reflection positivity in terms of 
the coupling constants 
$J{}^{a_1, \ldots, a_k}_{i_{1}, \ldots, i_{k}}$.
Again, the condition involves only the couplings
across the reflection plane.
 
In \S\ref{sec:reflectionsandgauge} we analyze different reflections 
$\Theta$ and $\Theta' = \alpha \Theta \alpha^{-1}$,
both of which interchange the same $\A_{\pm}$.
If they are related by a 
reflection-invariant gauge automorphism $\alpha$,
then
our characterization of reflection positivity
applies to $\Theta'$ as well as $\Theta$.

In  \S \ref{sec:exspin}
we illustrate the main results by showing that a number of spin Hamiltonians
with nearest neighbor as well as long-range interactions
fit naturally into our general framework. 
The central point of these examples is that 
our characterization of reflection positivity in Theorem \ref{Theorem:Summary}  can be applied easily to realistic physical systems. 

Reflection positivity of functionals has a long history in physics, as well as mathematics.  On the one hand, reflection positivity gives the relation between classical systems and quantum theory.  Furthermore reflection positivity is central in proving the existence of phase transitions/multiple equilibrium states in a number of classical and quantum systems. 
Concrete examples include, among many others, classical and quantum Heisenberg antiferromagnets,
hard-core nearest-neighbor and Coulomb lattice gases, and $\phi^4_{2}$ quantum fields.
Some earlier work can be found in \cite{OS1,OS2,GlimmJaffeSpencer-PhaseTransitions1975,FroehlichSimonSpencer,DysonLiebSimon1978,FroehlichIsraelLiebSimon1978,KleinLandau1981,FroehlichOsterwalderSeiler1983,Lieb1994,MacrisNachtergaele1996,NeebOlafsson2014,NeebOlafsson2015}. 

The present work was inspired by \cite{Jaffe-Pedrocchi2015-1,Jaffe-Pedrocchi2015-2}, and generalizes that work: here 
we obtain reflection positivity for couplings that are not necessarily diagonal 
(including long-range interactions),
for observables that are not necessarily even, and with hypotheses that are necessary as well as sufficient.

\subsection{Reflections}\label{sec:latsec}
Here we study a finite set $\Lambda$ which is an index set for the generators $c_{i}$ of our algebra. We assume that $\Lambda$ is invariant under an involution $\vartheta \colon \Lambda \to \Lambda$ that we call a {\em reflection}.   
We assume that $\vartheta$ exchanges two subsets $\Lambda_{\pm}$ whose union is $\Lambda$, and that $\vartheta$ has no fixed points.  

In specific models, $\Lambda$ is often a finite subset of a manifold 
$\calm$,  and $\vartheta$ is the restriction to $\Lambda \subset \calm$ of a reflection $\vartheta_{\calm} \colon \calm \rightarrow \calm$.
In the examples of interest, $\calm$ is a disjoint union
$\calm = \calm_{+} \sqcup \calm_{0} \sqcup \calm_{-}$,
where $\vartheta_{\calm}$ interchanges $\calm_{+}$
and $\calm_{-}$, and leaves the hypersurface 
$\calm_{0}$ invariant.
The set $\Lambda_{+}$ is then a finite set of points in $\calm_{+}$,
and $\Lambda_{-}$ is its reflection. 
 
We give a number of examples of this situation, where
$\calm$ is the Euclidean space $\R^{d}$, a torus $\T^{d}$, or a Riemann surface.

If $\calm = \R^{d}$, the 
reflection $\vartheta_{\R^{d}} \colon \R^{d} \rightarrow \R^{d}$
is given in suitable coordinates by 
\[\vartheta(x_0, x_1, \ldots, x_{d-1}) = (-x_0, x_1, \ldots, x_{d-1})\,.\]
The half-spaces $\R^{d}_{\pm} = \{ x \in \R^{d} \, : \, \pm x_{0} \geq 0\}$ have as a common boundary the 
reflection plane
$\R^{d}_{0} = \{ x \in \R^{d} \, : \, x_{0} = 0\}$.
Then $\Lambda_{+}\subset \R^{d}_{+}$ is a finite set of points on 
one side of the reflection plane $\R^{d}_{0}$, the set 
$\Lambda_{-}$ is its reflection, 
and $\Lambda = \Lambda_{+} \sqcup \Lambda_{-}$. 
Note that
$\Lambda$ contains no points in the reflection 
hyperplane.

An important example is  
the $d$-dimensional simple cubic lattice 
	\[
	\Lambda^{\mathrm{cubic}} 
		= \{-L-\half, -L + \half, \ldots, L-\half, L+\half\}^{d}\,,
	\]
with the reflection plane illustrated by the dashed line 
in  Figure~\ref{fig:cubic}.

\begin{figure}[h!]
   \centering
   \resizebox{0.75\textwidth}{!}{%
  		\includegraphics[width = 6 cm]{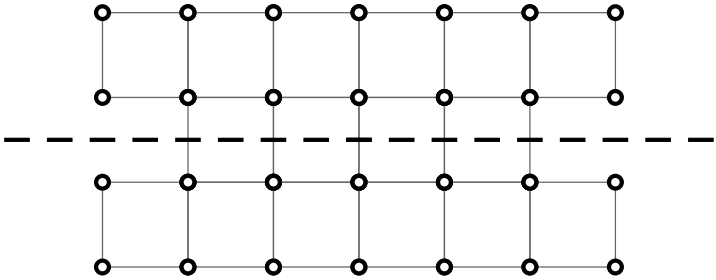}
		}
   \caption{Reflection in a cubic lattice.}
   \label{fig:cubic}
\end{figure}
Another example, with $\calm = \R^{2}$, is the  honeycomb lattice in Figure~\ref{fig:honeycomb}.
\begin{figure}[h!]
   \centering
      \resizebox{0.75\textwidth}{!}{%
    \includegraphics[width = 6 cm]{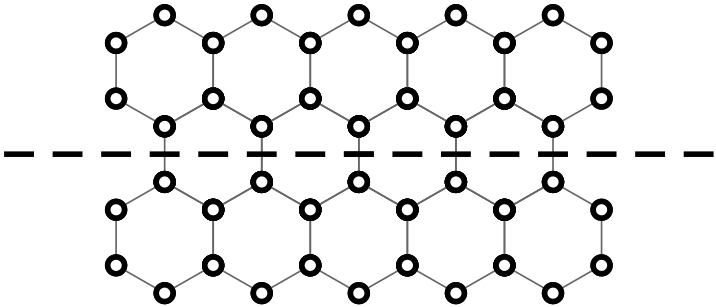}
   		}
   \caption{Reflection in the 2-dimensional honeycomb lattice.}\label{fig:honey}
   \label{fig:honeycomb}
\end{figure}

One often has periodic boundary conditions, 
in which case $\calm$ is
the torus $\T^{d}$ instead of $\R^{d}$.
The invariant hypersurface
$\calm_{0}$ is then the union of two $(d-1)$-tori.

Examples where $\calm$ is a Riemann surface of arbitrary genus
 arise from considering the conformal inversion $\vartheta$ of a Schottky double of an open Riemann surface $T$, with $\Lambda_{+}$ a finite set of points in $T$. 

In \S\ref{Sect:Introduction}--\S\ref{sec:RPSpinSyst}, including the main Theorems 
\ref{Thm:ReflectionPositivityH-I},
\ref{Thm:ReflectionPositivityH}, \ref{Theorem:RPGibbsFunctionals}, and 
\ref{thm:maincorollarySpinSyst}, we only need $\Lambda$ to be an abstract set with a fixed point free involution $\vartheta$.
In the discussion of examples in  \S\ref{sec:exspin}, we require additional structure for  
$\Lambda$, involving its geometric significance as a subset of a manifold $\calm$, as explained later.  
\subsection{Majoranas}\label{sec:majgau}
One defines an algebra of Majoranas
 on the lattice $\Lambda$ as  the  $*$-algebra $\A$ with self-adjoint generators $c_{i}=c_{i}^{*}$ that satisfy the Clifford relations \eqref{eq:Cliffrels}. 
For any subset $\Gamma\subset\Lambda$, let  $\A(\Gamma)$ denote the algebra generated by the $c_{j}$'s with 
$j\in\Gamma$.  In particular, $\A=\A(\Lambda)$, and we define $\A_{\pm}:=\A(\Lambda_{\pm})$.  

We call the automorphism  $\alpha \colon \A \rightarrow \A$ that implements the $\Z_{2}$ grading a {\em global gauge automorphism}. 
On the generators, it satisfies 
	\be\label{eq:GlobalGauge}
		c_{j} \mapsto \alpha(c_{j})= -c_{j}\;.
	\ee
The algebra $\A$ decomposes into the 
spaces $\A^{\even}$ and $\A^{\odd}$ of elements that are 
even and odd for the $\Z_2$-grading,
\[
	\A = \A^{\even} \oplus \A^{\odd}\,.
\] 
In the same vein, $\A(\Gamma) = \A(\Gamma)^{\even}\oplus \A(\Gamma)^{\odd}$.
An element $A\in \A$ that is either even or odd is called \emph{homogeneous}. 
Since $A\in \A$ is even if $\alpha(A) = A$ and odd if $\alpha(A) = -A$, 
the even elements are also called
\emph{globally gauge invariant}.

Define the \emph{degree} $\abs{A}$ of $A$ as  
$\abs{A} = 0$
for $A \in \A^{\even}$, and $\abs{A} = 1$ for $A \in \A^{\odd}$.
The algebra $\A(\Gamma)$ commutes with 
$\A^{\even}(\Gamma')$ when $\Gamma\cap\Gamma'=\varnothing$.  
More generally, if $A\in \A(\Gamma)$ and $B \in \A(\Gamma')$ are both eigenvectors of $\alpha$, then
\[
	AB = (-1)^{\abs{A}\abs{B}}BA\;,
	\qsp{when} 
	 \Gamma\cap\Gamma'=\varnothing\;.
\]
One says that $\A(\Gamma)$ and $\A(\Gamma')$ supercommute if $\Gamma$ and $\Gamma'$
are disjoint. 

\subsection{Reflections and Invariant Bases}\label{sec:bases}
The reflection $\vartheta \colon \Lambda \rightarrow \Lambda$
defines an \emph{anti-linear} *-automorphism 
$\Theta \colon \A \rightarrow \A$ given by 
\be\label{eq:ReflectionC}
\Theta(c_{i_1} \cdots c_{i_{k}}) := c_{\vartheta(i_1)} \cdots c_{\vartheta(i_{k})}\,.
\ee
Note that $\Theta$ exchanges 
$\A_{+}$ with $\A_{-}$, namely $\Theta(\A_{\pm})=\A_{\mp}$, and satisfies $\Theta^2 = \mathrm{Id}$.
We construct bases of $\A$ that are adapted to this reflection.

For $\Gamma \subseteq \Lambda$,
let $\cS_{\Gamma}$ denote the set of sequences
$\I = (i_{1}, \ldots, i_{k})$ of \emph{distinct} 
lattice points $i_{1}, \ldots, i_{k} \in \Gamma$. 
For the important choices $\Gamma = \Lambda$, 
$\Gamma = \Lambda_{+}$ and $\Gamma = \Lambda_{-}$,
we denote $\cS_{\Gamma}$ by $\cS$, $\cS_{+}$, and $\cS_{-}$,
respectively.
For $\I \in \cS_{\Gamma}$, define the monomial
\[
	C_{\I} := c_{i_{1}}\cdots c_{i_{k}}\,,
\]
and define $C_{\I} := I$ for $\I = \varnothing$.  Each $C_{\I}$ is an eigenvector of the gauge automorphism $\alpha$, and we denote its degree by
\be\label{eq:degree}
	\abs{\I} := \abs{C_{\I}}\,.
\ee
Then $\abs{\I} = 0$ if $k$ is even, and $\abs{\I} = 1$
if $k$ is odd.  Also
	\be\label{eq:Adjoint C}
		C_{\I}^{*} = (-1)^{\frac12 k(k-1)} C_{\I}\;.
	\ee

The algebra $\A(\Gamma)$ is spanned by the operators 
$C_{\I}$ with $\I \in \cS_{\Gamma}$,
but they are linearly dependent.
In fact, $C_{\I} = \pm C_{\I'}$ if the sets
$\{i_{1}, \ldots, i_{k}\}$ and 
$\{i'_{1}, \ldots, i'_{k'}\}$ are the same.
A choice 
$\cP_{+} \subseteq \cS_{+}$ 
such that every set $\{i_{1}, \ldots, i_{k}\}$
of distinct lattice points
corresponds to precisely one tuple
$(i_{1}, \ldots, i_{k})$ in $\cP_{+}$ yields a 
basis
\[
\B_{+} = \{C_{\I}\,;\, \I \in \cP_{+}\}
\]
 of $\A_{+}$.
 This, in turn, 
yields a basis 
$\B_{-} = \Theta(\B_{+})$
of $\A_{-}$.

\subsection{The Twist}

From the two bases $\B_{+}$ and $\B_{-}$, we construct a basis 
$\B$ of $\A$, that  is adapted to the reflection $\Theta$.
For this, fix a square root of minus one, 
	\be\label{zeta}
	\zeta = \pm \sqrt{-1}\;,
	\ee 
and define a basis for $\A$ by
\[
\B := \{\zeta^{\abs{\I}\abs{\I'}}\Theta(C_{\I}) C_{\I'}\,;\; C_{\I}, C_{\I'} \in \B_{+} \}\,.
\]
Although the main results on reflection positivity will hold for both twists $\zeta = \sqrt{-1}$
and $\zeta = -\sqrt{-1}$, the class of allowed Hamiltonians will \emph{not} be the same.

Note that, in a sense, the basis elements in $\B$ are the geometric mean of the operators
$\Theta(C_{\I}) C_{\I'}$ and $C_{\I'}\Theta(C_{\I})$, which differ by a factor 
$(-1)^{\abs{\I}\abs{\I'}}$.
The identity $I = C_{\varnothing} = \Theta(C_{\vn}) = \Theta(C_{\varnothing})C_{\varnothing}$ 
is a basis element in  all three bases $\B_+$, $\B_{-}$ and $\B$.
Every $A\in \A$ has an expansion
	\be\label{eq:basisexpansion-0}
		A 
		= \sum_{\I,\I'} a_{\I\I'}\, \zeta^{\abs{\I}\abs{\I'}}
			\Theta(C_{\I}) C_{\I'}\,,
	\ee
which is unique if the $\I, \I'$ are restricted to be in $\cP_{+}$.

\subsection{Twisted Product}\label{Sect:TwistedProduct}
In order to streamline notation, introduce the following (non-associative) twisted product
	\be
	\cz:\; \A\times \A\rightarrow \A\;.
	\ee 	

\begin{defini}\label{def:twistprod}
Let $A \in \A$ be of the form $A = A_{-}A_{+}$ 
with $A_{\pm} \in \A_{\pm}$, and similarly  
$B = B_{-}B_{+}$ with $B_{\pm} \in \A_{\pm}$.
If $A_{\pm}$ and $B_{\pm}$ are homogeneous, then 
$A\circ B$ is defined by
\[
A \circ B := 
\zeta^{\abs{A_{-}} \abs{B_{+}} - \abs{A_{+}} \abs{B_{-}}} 
AB\,.
\]
This extends bilinearly to the product
$\cz:\; \A\times \A\rightarrow \A$. 
\end{defini}
Note that the formula 
$
	A \circ B := 
	\zeta^{\abs{A_{-}} \abs{B_{+}} - \abs{A_{+}} \abs{B_{-}}} AB
$
also holds for $A = A_{+}A_{-}$ and $B = B_{+}B_{-}$. 
One finds
\be\label{eq:twistonbasis}
X_{\I_{1}\I'_{1}} \circ X_{\I_{2}\I'_{2}}
=
\zeta^{\abs{\I_{1}}\,\abs{\I'_{2}} - \abs{\I'_{1}}\,\abs{\I_{2}}} 
X_{\I_{1}\I'_{1}} X_{\I_{2}\I'_{2}}
\ee
for twisted products of elements of the form 
$X_{\I \I'} = \Theta(C_{\I})C_{\I'}$ or
$X_{\I \I'} =  C_{\I'} \Theta(C_{\I})$.

In terms of the twisted product, the basis $\B$ can be written
\be
	\B = \{ \Theta(C_{\I})\cz C_{\I'}\,:\,\I,\I' \in\cP_{+}\}\,. 
\ee
Correspondingly we can rewrite the expansion \eqref{eq:basisexpansion-0} of a general element $A\in \A$ in basis elements as 
	\be\label{eq:basisexpansion}
		A = \sum_{\I,\I'\in\cP_{+}}
		a_{\I\I'} \,\Theta(C_{\I})\circ C_{\I'}\;.
	\ee

The twisted product has a number of
useful properties.
For example $\A_{+}$ and $\A_{-}$ commute with respect to the twisted product.
\begin{prop}\label{prop:supercomm}
If $A_{+} \in \A_{+}$ and $B_{-} \in \A_{-}$, then
$A_{+}\circ B_{-} = B_{-} \circ A_{+}$.
\end{prop}
\begin{proof}
It suffices to prove this for homogeneous elements, in which case
the result
follows from 
$A_{+} \circ B_{-} = \zeta^{- \abs{A_{+}} \abs{B_{-}}} A_{+}B_{-}$, 
$B_{-} \circ A_{+} = \zeta^{\abs{A_{+}} \abs{B_{-}}} B_{-}A_{+}$,
and $A_{+}B_{-} = (-1)^{\abs{A_{+}} \abs{B_{-}}} B_{-}A_{+}$.
\end{proof}

The twisted product respects the reflection:

\begin{prop}\label{prop:refresp}
For all $A, B \in \A$, one has $\Theta(A \circ B) = \Theta(A)\circ \Theta(B)$.
\end{prop}
\begin{proof}
It suffices to check this for $A = A_{-}A_{+}$ and $B = B_{-}B_{+}$
as in Definition \ref{def:twistprod}.
By antilinearity of $\Theta$, one then finds
\[
\Theta(A\circ B) = \Theta(\zeta^{\abs{A_{-}} \abs{B_{+}} - \abs{A_{+}} \abs{B_{-}}} AB)=
\zeta^{-\abs{A_{-}} \abs{B_{+}} + \abs{A_{+}} \abs{B_{-}}} \Theta(A)\Theta(B)
\]
for the left side of the equation. For the right side, one finds the same expression
\beqs
\Theta(A)\circ \Theta(B) &=& 
\Theta(A_{-})\Theta(A_{+}) \circ \Theta(B_{-})\Theta(B_{+})\\
&=&
\zeta^{\abs{A_{+}} \abs{B_{-}} -\abs{A_{-}} \abs{B_{+}}}
\Theta(A)\Theta(B)\,,
\eeqs
since $\Theta(A_{\pm}), \Theta(B_{\pm}) \in \A_{\mp}$.
\end{proof}

It follows that the reflection permutes the basis $\B$.  
\begin{cor}\label{cor:refbasis}
The twisted product satisfies
	\be\label{eq:ThetaReverse} 
		\Theta \big( \Theta(A)\cz B \big) 
		= \Theta(B)\cz A\,,
		\quad\text{for} \quad
		A, B \in \A_{+}\;, \text{  or  } A,B \in \A_{-}\;.
	\ee
In particular, the basis $\B$ is permuted by $\Theta$, 
	\be\label{Theta-Basis}
	\Theta(\Theta(C_{\I})\cz C_{\I'}) = \Theta(C_{\I'})\cz C_{\I}\;.
	\ee 
\end{cor}
\begin{proof}
By Proposition \ref{prop:refresp}, one has $\Theta(\Theta(A) \circ B) = A \circ \Theta(B)$,
which equals  $\Theta(B)\circ A$
by Proposition \ref{prop:supercomm}.
\end{proof}

\begin{prop}\label{prop:reflectioninv}
Let $A\in\A$ have the expansion \eqref{eq:basisexpansion}. Then  $A$ is reflection invariant, namely $\Theta(A) = A$, if and only if the matrix $a_{\I \I'}$ is hermitian, namely $a_{\I'\I} = \overline{a_{\I\I'}}$.
\end{prop}

\begin{proof}
This follows from anti-linearity of $\Theta$ and Corollary \ref{cor:refbasis}.
\end{proof}

Define $k \colon \Lambda \rightarrow \mathbb{N}$ 
by $k_{\I} = r$ for $\I = (i_1, \ldots, i_{r})$. 
When dealing with adjoint operators, one 
frequently encounters the derived expressions
\be\label{eq:q}
q_{\I} := (-1)^{\frac{1}{2}k_{\I} (k_{\I}-1)} \quad\text{and}\quad
s_{\I} := \zeta^{\frac{1}{2}k_{\I}(k_{\I}-1)}\,.
\ee
Note that $q_{\I}$ is 4-periodic in $k$, and $s_{\I}$ is 8-periodic.

\begin{prop}\label{Prop:Adjoints}
Let $\I,\I'\in \cP_{+}$.  Then  
	\beq\label{eq:Early Identity-0}
		(\Theta(C_{\I})\cz C_{\I'})^* 
		&=& \label{eq:Basis-Adjoint-0}
			q_{\I}\,q_{\I'}\,
			\Theta(C_{\I})\cz C_{\I'}\;.
	\eeq
\end{prop}

\begin{proof}
As $\Theta$ is a $*$-automorphism, 
	\[
	(\Theta(C_{\I})\cz C_{\I'})^* 
	= \zeta^{-\abs{\I}\abs{\I'}}
(\Theta(C_{\I}) C_{\I'})^{*}
		= \zeta^{-\abs{\I}\abs{\I'}}
C_{\I'}^{*}\Theta(C^{*}_{\I}) \;.
	\]
Inserting \eqref{eq:Adjoint C} gives
\beqs
	(\Theta(C_{\I})\cz C_{\I'})^*  
	&=&
	\zeta^{-\abs{\I}\abs{\I'}}\,q_{\I}\,q_{\I'}\,C_{\I'}\Theta(C_{\I})\\
	&=& 
	\,q_{\I}\,q_{\I'}\,\Theta(C_{\I})\cz C_{\I'}\,.
\eeqs
In the last equality we use $C_{\I'}\Theta(C_{\I}) = \zeta^{2\abs{\I} \abs{\I'} }\Theta(C_{\I})C_{\I'}$ and the definition of the circle product to give the desired relation \eqref{eq:Basis-Adjoint-0}.  
\end{proof}

Using this, one derives the following characterization of hermiticity.
\begin{cor}
If $A \in \A$ has an expansion \eqref{eq:basisexpansion} with coefficients 
$a_{\I\I'}$, then $A^*$ has coefficients $q_{\I}\,q_{\I'}\,\overline{a_{\I\I'}}$.
The operator $A$ is hermitian if and only if  
$s_{\I}\,s_{\I'}\,a_{\I\I'}$ is real for all $\I,\I'\in \cP_{+}$.
\end{cor}
\begin{proof}
The first statement follows from
	\beq
		A^{*}
		&=&  \sum_{\I,\I'\in\cP_{+}}
		\overline{a_{\I\I'}} \lrp{\Theta(C_{\I})\circ C_{\I'}}^{*}\nn
		&=&   \sum_{\I,\I'\in\cP_{+}}
		\overline{a_{\I\I'}} \,q_{\I}\,q_{\I'}\,
		%(-1)^{\frac{1}{2}k(k-1)+\frac{1}{2}k'(k'-1)}\,
		\Theta(C_{\I})\cz C_{\I'}\;.
	\eeq
Therefore, $A$ is hermitian if and only if 
$a_{\I\I'} = \overline{a_{\I\I'}}\,q_{\I}\,q_{\I'}$.
Since $s^2_{\I} = s^{-2}_{\I} = q_{\I}$, 
this is equivalent to $s_{\I}\,s_{\I'}\,a_{\I\I'} = \overline{s_{\I}\,s_{\I'}\,a_{\I\I'}}$.
\end{proof}

\subsection{The Tracial State}
Define the functional  $\Tr:\A \to\C$  by
\be\label{eq:deftr}
\Tr(A) = a_{\vn\vn}\,,
\ee
where $a_{\I \I'}$ are the coefficients in \eqref{eq:basisexpansion}.

\begin{prop}\label{Prop:Orthogonality}
Let $\I_{0}, \I_{1}$ and $\I'_{0}, \I'_{1}$ be elements of $\cP_{+}$. Then
\be \label{eq:traces1}
\Tr\Big(\lrp{\Theta(C_{\I_0})\cz C_{\I'_{0}}}^* {\Theta(C_{\I_{1}})\cz C_{\I'_{1}} }\Big)
= \delta_{\I_0 \I_{1}} \delta_{\I'_{0} \I'_{1}}\,.
\ee
Also  
\be \label{eq:traces2}
\Tr\lrp{\lrp{\Theta(C_{\I_0})\cz C_{\I'_{0}} }\lrp{ \Theta(C_{\I_{1}})\cz C_{\I'_{1}}} }
= q_{\I_{0}} \,q_{\I'_{0}}\,\delta_{\I_0 \I_{1}} \delta_{\I'_{0} \I'_{1}}  \,.
\ee
\end{prop}

\begin{proof}
The identity \eqref{eq:traces2} is equivalent to \eqref{eq:traces1} as a consequence of \eqref{eq:Basis-Adjoint-0}.  
The left hand side of \eqref{eq:traces1} vanishes
unless $\I_{0} = \I_{1}$ and $\I'_{0} = \I'_{1}$, in which
case \eqref{eq:Early Identity-0} along with $C^{*}_{\I_{1}} C_{\I_{1}} = C^{*}_{\I'_{1}} C_{\I'_{1}} = I$ 
yields 
\beq
&&\hskip -.5in
(\Theta(C_{\I_0})\cz C_{\I'_{0}})^* \,\cdot\, \Theta(C_{\I_{1}})\cz C_{\I'_{1}}\nn
&&= \zeta^{-\abs{\I_{0}}\abs{\I'_{0}}+\abs{\I_{1}}\abs{\I'_{1}}}
C_{\I'_{0}}^{*}\Theta(C^{*}_{\I_0})
\Theta(C_{\I_{1}}) C_{\I'_{1}}
\nn
&&=
C_{\I'_{1}}^*\Theta(C_{\I_{1}}^*C_{\I_{1}})C_{\I'_{1}} = I\,.
\eeq
This proves equation \eqref{eq:traces1}.  
\end{proof}

\begin{prop}[\bf The Normalized Trace]  \label{prop:propertiesTr}
The functional $\Tr$
is a tracial, factorizing, reflection-invariant state.  Namely 
\begin{itemize}
\item[(a)]  It is  normalized, $\Tr(I) = 1$. 
\item[(b)] It is positive definite, $\Tr(A^{*}A)  \geq 0$ for all $A\in\A$, with equality only for $A=0$.
\item[(c)] It is cyclic, 
	\be
		\Tr(AB) = \Tr(BA) \qsp{for all} A,B\in \A\;.
	\ee
\item[(d)]  It satisfies
	\be\label{eq:Trace-Reflection}
		\Tr(\Theta(A)) = \overline{\Tr(A)}\qsp{for all} A\in\A\;.
	\ee
\item[(e)]  It factorizes,
\be\label{factorization}
	\Tr(A_{-}A_{+}) = \Tr(A_{-})\Tr(A_{+})\;,
	\qsp{for} A_{\pm} \in \A_{\pm}\;.
\ee
\end{itemize}
\end{prop}

\begin{proof}
(a) As $I=\Theta(C_{\vn})\circ C_{\vn}$, one has $\Tr(I)=1$.

(b) From \eqref{eq:traces1} and the expansion \eqref{eq:basisexpansion}, one finds 
	\[
	\Tr(A^{*}A) = \sum_{\I,\I' \in \cP_{+}} |a_{\I\I'}|^2 \geq 0\;.
	\]
Furthermore $\Tr(A^{*}A) =0$ only if all the $a_{\I\I'}=0$.  As the $\Theta(C_{\I})\circ C_{\I'}$ are a basis, the vanishing of $a_{\I\I'}$ ensures that $A=0$.  Hence $\Tr$ is positive definite.

(c) From equation \eqref{eq:traces2}, one obtains
\be\label{eq:Trace-Of-Product}
\Tr(AB) = \Tr(BA) = \sum_{\I,\I' \in \cP_{+}} q_{\I} q_{\I'} a_{\I\I'}b_{\I\I'}\,. 
\ee
Hence the state $\Tr$ is cyclic.

(d) As $\Theta$ is antilinear and the basis elements satisfy \eqref{Theta-Basis}, it follows that  $\Tr$ satisfies \eqref{eq:Trace-Reflection}.

(e) To demonstrate factorization, consider  $A_{-} = \sum_{\I \in \cP_{+}} a_{\I\vn} \,\Theta(C_{\I})$ and $B_{+} = 
\sum_{\mathfrak{K}' \in \cP_{+}} b_{\vn\mathfrak{K}'}\,C_{\mathfrak{K}'}$.  In this case, identity \eqref{eq:Trace-Of-Product} takes the form $\Tr(A_{-}B_{+}) = a_{\vn\vn}b_{\vn\vn}=\Tr(A_{-})\,\Tr(B_{+})$.  So the factorization property follows.
\end{proof}

\begin{cor}\label{cor:Zisreal}
If $H \in \A$ is reflection invariant, $\Theta(H) = H$, then the partition sum
$Z_{H} = \tr(e^{-H})$ is real.
\end{cor}
\begin{proof}
Since $\Theta$ is an automorphism, it follows from $\Theta(H) = H$ that 
$\Theta(e^{-H}) = e^{-H}$.
Using Proposition \ref{prop:propertiesTr}.d, one then finds
\[Z_{H} = \Tr(e^{-H}) = \Tr(\Theta(e^{-H})) = \overline{\Tr(e^{-H})} =  \overline{Z_{H}}\,,\]
so that $Z_{H}$ is real.
\end{proof}

\section{Reflection Positive Functionals}\label{Sect:Reflection-InvariantFunctionals}
In this section, we characterize reflection invariance and 
reflection positivity of linear functionals in terms of their
density matrix.

\subsection{Reflection Invariance}
Let $\omega \colon \A \rightarrow \C$ be a linear functional on $\A$.
From  Proposition \ref{prop:propertiesTr}.b,  we infer that every functional can be written 
	\be\label{General Boltzmann Functional}
	\omega(A) 
	= \Tr(A R)
	\ee
for a unique \emph{density matrix} $R \in \A$.  
If $\omega$ is a state, then $R$ is a positive operator with trace 1.

Consider the sesquilinear form $\lra{\,\cdot\,,\,\cdot\,}_{\rz  }$ on $\A$ given as 
	\be\label{eq:rpform}
		\lra{A,B}_{\rz  }
		:= \omega(\Theta(A)\cz B) 
		= \Tr((\Theta(A)\cz B)R)\,.
	\ee
If we expand $R$ in terms of matrix elements $r_{\I\I'}$ as
	\be\label{eq:R}
	R = \sum_{\I, \I'\in\cP_{+}} r_{\I \I'} \,\Theta(C_{\I})\cz C_{\I'}\,,
	\ee
then  \eqref{eq:Adjoint C} and Proposition \ref{Prop:Orthogonality} ensure that 
	\be\label{eq:MatrixElements-R}
		r_{\I\I'}
		= \lra{C_{\I}^{*}, C_{\I'}^{*}}_{\rz}\;,
		\quad\text{where}\quad
		C_{\I}, C_{\I'}\in\B_{+}\;.
	\ee

\begin{defini}[\bf Reflection Invariance]
The linear functional $\omega$ is reflection invariant on $\A$ if $\omega(\Theta(A)) = \overline{\omega(A)}$ for all $A \in \A$.
\end{defini}

\begin{prop}[\bf Reflection-Invariant Functionals]
\label{prop:Hermiticity}
The following conditions are equivalent:
\begin{itemize}
\item[(a)] The functional $\omega(A) = \Tr(AR)$ is reflection invariant on $\A$.

\item[(b)] The  operator $R$ is reflection invariant, $\Theta(R) = R$.

\item[(c)] The matrix $r_{\I \I'}$ is hermitian, 
	$r_{\I'\I} = \overline{r_{\I\I'}}$.
\item[(d)] The sesquilinear form $\lra{\,\cdot\,,\,\cdot\,}_{\rz}$
is hermitian on $\A_+$,
\[
  \lra{A,B}_{\rz} = \overline{\lra{B,A}}_{\rz}\;,
  \quad\text{for all}\quad
  A,B\in\A_{+}\,.
 \] 
 \end{itemize}
\end{prop}

\begin{proof}
(b)$\Rightarrow$(a):
By Proposition \ref{prop:propertiesTr}.d, the trace is reflection invariant, 
$\Tr(\Theta(X)) = \overline{\Tr(X)}$. If $\Theta(R)=R$, one finds
\[
\overline{\Tr(AR)} = \Tr(\Theta(AR)) = \Tr(\Theta(A)\Theta(R)) = \Tr(\Theta(A)R)\,.
\]
Thus $\overline{\omega(A)} = \omega(\Theta(A))$, and $\omega$ is reflection invariant.

(a)$\Rightarrow$(d):
If $\omega$ is reflection invariant, then 
\[
\overline{\omega(\Theta(B)\circ A)} 
= \omega(\Theta(\Theta(B)\circ A)) = \omega(\Theta(A)\circ B)\,,
\]
where the second equality follows from Proposition \ref{cor:refbasis}. 

(d)$\Rightarrow$(b):
Since $\overline{\lra{B,A}}_{\rz} = \overline{\Tr((\Theta(B)\circ A) R)}$, 
reflection invariance of the trace and Proposition\ \ref{cor:refbasis} yield
\[
 \overline{\lra{B,A}}_{\rz}
= \Tr(\Theta(\Theta(B)\circ A) \Theta(R))
= \Tr((\Theta(A)\circ B) \Theta(R))\,,
\]
for all $A, B \in \A_{+}$.
Since $\lra{A,B}_{\rz} = \Tr((\Theta(A)\circ B) R)$, 
we infer from 
$\lra{A,B}_{\rz} = \overline{\lra{B,A}}_{\rz}$
that
\[
\Tr((\Theta(A)\circ B) R)
=
\Tr((\Theta(A)\circ B) \Theta(R))\,.
\]
Since $\A$ is spanned by elements of the form $\Theta(A)\circ B$
with $A, B \in \A_{+}$, nondegeneracy of the trace implies $\Theta(R) = R$. 

We conclude that (a)$\Leftrightarrow$(b)$\Leftrightarrow$(d).
The equivalence (b)$\Leftrightarrow$(c) was already 
proven in Proposition~\ref{prop:reflectioninv}.
\end{proof}

A linear functional $\omega \colon \A \rightarrow \C$ is called \emph{even} if $\omega(\A^{\rm odd}) = \{0\}$.
Note that if $R$ is even, then also $\omega(A) = \Tr(RA)$ is even.

\begin{prop}\label{evenform}
If $\omega$ is even, then
$\A^{\rm even}_{+}$ and $\A^{\rm odd}_{+}$ are orthogonal,
\[\lra{\A^{\rm even}_{+}, \A^{\rm odd}_{+}}_{\rz} = \{0\}\,.\] 
\end{prop}
\begin{proof}
For $A \in \A_{+}^{\rm even}$ and $B \in \A^{\rm odd}_{+}$, 
one has $\lra{A, B}_{\rz} = \omega(\Theta(A)\circ B).$ This 
equals zero, as
$\Theta(A)\circ B \in \A^{\rm odd}$. 
\end{proof}

\subsection{Reflection Positivity} \label{sect:Reflection-Positivity}
In this section, we characterize reflection positive functionals in terms 
of their density matrix.

\begin{defini}  \label{Defn:RP}
The linear functional $\omega$ in \eqref{General Boltzmann Functional} is reflection positive on $\A_{+}$  with respect to  $\Theta$,  if the form $\lra{\,\cdot\,,\,\cdot\,}_{\rz  }$ in \eqref{eq:rpform} is positive, semidefinite on $\A_{+}$.
\end{defini}

The reflection positive Hilbert space $\cH$ is defined as the completion with respect to $\lra{\,\cdot\,,\,\cdot\,}_{\rz}$ of 
the quotient of $\A_{+}$ by the null space. If $\omega$ is even, $\cH$ will be a \emph{super Hilbert space}
in the following sense.  

\begin{defini}[\cite{SuperSolutions}, \S 4.4]
A super Hilbert space is a $\Z_2$-graded vector space $\cH = \cH^{\rm even} \oplus \cH^{\rm odd}$
with a form $(\, \cdot \,,\,\cdot\,) \colon \cH \times \cH \rightarrow \C$ that is 
\begin{itemize}
\item[-] linear in the second argument,
\item[-] graded symmetric, 
\[(w, v) = (-1)^{\abs{v}\abs{w}}\overline{(v,w)}
\] 
for $v,w \in \cH$ homogeneous
\item[-] even, $(v, w) = 0$ for $v \in \cH^{\rm even}$ and $w \in \cH^{\rm odd}$
\item[-] positive, in the sense that
\begin{eqnarray}\label{eq:superpositive2}
0 &< & (v,v)\phantom{z} \qsp{for} 0 \neq v \in \cH^{\rm even}\\
0 &<& \zeta (v,v) \qsp{for} 0 \neq v \in \cH^{\rm odd}\,. \nonumber
\end{eqnarray}
\end{itemize}
Furthermore, the total space $\cH$ is required to be complete for the 
scalar product defined by  
$\lra{v,w} := (v,w)$ for $v,w \in \cH^{\rm even}$,
$\lra{v,w} := \zeta (v,w)$ for $v,w \in \cH^{\rm odd}$,
 and $\lra{v,w} := 0$ for $v\in \cH^{\rm even}$, $w \in \cH^{\rm odd}$.
\end{defini}

\begin{prop}
If $\omega$ is even, reflection invariant, and reflection positive, then the completion $\cH$ with respect to
$\lra{\,\cdot\,, \,\cdot\, }_{\rz}$
of $\A_{+}$ modulo the null space 
is a super Hilbert space with 
the form
\[(A, B) := \omega(\Theta(A)B).\]
\end{prop}

\begin{proof}
The form $\lra{A,B}_{\rz} = \omega(\Theta(A)\circ B)$ is 
hermitian by Proposition \ref{prop:Hermiticity}, and positive semidefinite 
by reflection positivity of $\omega$.
As $\omega$ is even, Proposition \ref{evenform} yields $\cH^{\rm even} \perp \cH^{\rm odd}$.
Since $\lra{A,B}_{\rz} = (A,B)$ for $A, B \in \A_{+}^{\rm even}$ and 
$\lra{A,B}_{\rz} = \zeta (A,B)$ for $A, B \in \A_{+}^{\rm odd}$, 
graded symmetry of $(\,\cdot\,,\,\cdot\,)$ follows from hermiticity
of $\lra{\,\cdot\,,\,\cdot\,}_{\rz}$, and positivity of $(\,\cdot\,,\,\cdot\,)$ 
(equation \eqref{eq:superpositive2}) follows from the fact that 
$\lra{\,\cdot\,,\,\cdot\,}_{\rz}$ is positive semidefinite.
\end{proof}

As $\cH^{\rm even} \perp \cH^{\rm odd}$, the value of $A \circ B$ for $A \in \A^{\rm even}$ and 
$B \in \A^{\rm odd}$ is quite immaterial for even functionals; 
the relevant property of the twisted product 
is that $A \circ B = AB$ for $A, B \in \A^{\rm even}_{+}$, and 
$A \circ B = \zeta AB$ for $A, B \in \A^{\rm odd}_{+}$. 
Our choice for Definition \ref{def:twistprod} was merely motivated by the wish 
to treat $\A_{+}$ and $\A_{-}$ on equal footing.

\begin{prop}
The functional $\omega$ in \eqref{General Boltzmann Functional} is reflection positive on $\A_{+}$, if and only if  it is reflection positive  on $\A_{-}$.  In fact
	\be\label{eq:RP-Agree}
	\lra{\Theta(A), \Theta(B)}_{\rz}
	=\lra{B,A}_{\rz}\;,
	\qsp{for} A,B\in\A_{+}\;.
	\ee
\end{prop}

\begin{proof}
For $A,B \in \A_{+}$, we infer from Proposition \ref{prop:supercomm} and Corollary \ref{cor:refbasis} that
\[
\omega(\Theta(\Theta(A))\circ \Theta(B))
=
\omega(A\circ \Theta(B))
=
\omega(\Theta(B)\circ A)\,.
\]
The first term equals $\lra{\Theta(A), \Theta(B)}_{\rz}$
and the last one $\lra{B,A}_{\rz}$.  
\end{proof}

\begin{thm}[\bf Basic Reflection Positivity]\label{Theorem:RPFunctionals}
The linear functional $\omega$ in \eqref{General Boltzmann Functional} is reflection positive on $\A_{+}$, if and only if the matrix $r_{\I \I'}$ defined in \eqref{eq:R} is positive semidefinite.
\end{thm}

\begin{proof}
Expand $A,B \in \A_{+}$ as $A = \sum_{\I\in \cP_{+}}a_{\I}\,C_{\I}$
and $B = \sum_{\I\in \cP_{+}}b_{\I}\,C_{\I}$.  Using \eqref{eq:Adjoint C} and \eqref{eq:MatrixElements-R} we obtain
	\beqs
	\lra{A,B}_{\rz  } 
	&=& \Tr((\Theta(A)\cz B)R) \\
	&=& 
	\sum_{{\I_{0},\I'_{0}\in \cP_{+}}\atop{\I_{1},\I'_{1}} \in \cP_{+}}
	\overline{a_{\I_{0}}}\,b_{\I'_{0}},r_{\I_{1}\I'_{1}}
	\Tr\Big(\lrp{\Theta(C_{\I_{0}})\cz C_{\I'_{0}}}
		(\Theta(C_{\I_{1}})\cz C_{\I'_{1}})\Big)\\
	&=& 
	\sum_{\I,\I'}
	\overline{a_{\I}}\,q_{\I} \,b_{\I'}q_{\I'}\, r_{\I \I'}\,.
\eeqs
It follows that $\lra{A,A}_{\rz} \geq 0$ for all $A \in \A_{+}$ if and only if 
the matrix $r_{\I\I'}$ is is positive semidefinite. 
\end{proof}

\section{Sufficient Conditions for Reflection Positivity}\label{Sect:Sufficient}
In statistical physics, Gibbs states are defined in terms 
of a Hamiltonian $H$, which in turn is given by a matrix $J$
of coupling constants.
In this section, we provide a sufficient condition on $J$ for the 
associated Gibbs state to be reflection positive.
This will be further refined to a necessary and sufficient condition 
in Section \ref{Sect:Which}.

\subsection{Density Matrices and Hamiltonians}
For a (not necessarily hermitian) Hamiltonian $H \in \A$, consider the unnormalized density matrix  $R=e^{-H}$.  We now focus on the Hamiltonian $H$ rather than $R$, and define the \emph{Boltzmann functional}
$\omega_{H} \colon \A \rightarrow \C$ by
	\be\label{eq:Boltzmann}
		\omega_{H}(A)= \Tr(A\,e^{-H})\,.
	\ee
If the partition function $Z_H:= \Tr(e^{-H})$ is nonzero, then define the \emph{Gibbs functional} 
$\rho_{H} \colon \A \rightarrow \C$  as the normalization of $\omega_{H}$,
	\be\label{eq:defequilibrium}
		\rho_{H}(A) := \frac{\omega_{H}(A)}{Z_H} =\frac{\Tr(Ae^{-H})}{\Tr(e^{-H})}\,.	
	\ee
	
Using equation \ref{eq:rpform}, the (unnormalized) Boltzmann functional $\omega_{H}$ yields the ses\-qui\-li\-near form
	\be\label{eq:Hsesq0}
		\lra{A,B}^{0}_{\hz} := \Tr((\Theta(A) \cz B)\, e^{-H})
	\ee
on $\A_{+}$.
Similarly, the (normalized) Gibbs functional $\rho_{H}$ yields the form
	\be \label{eq:Hsesq}
		\lra{A,B}_{\hz} := \frac{\Tr((\Theta(A) \cz B) \,e^{-H})}{\Tr(e^{-H})}\,.
	\ee

\begin{Remark}\rm
The functional $\omega_{H}$ in \eqref{eq:Boltzmann} is reflection positive on $\A_{+}$ if $\lra{A,B}^{0}_{\hz} $  in \eqref{eq:Hsesq0} is positive semidefinite on $\A_{+}$, 
and the  
functional $\rho_{H}$ defined in \eqref{eq:defequilibrium} is reflection positive on $\A_{+}$ if the   form  $\lra{A,B}_{\hz}$ in \eqref{eq:Hsesq} is positive semidefinite on $\A_{+}$.  
Note that $\omega_{H}$ and $\rho_{H}$ are even if $H$ is globally gauge invariant.
\end{Remark}
	
In \S\ref{sec:RPomega} we show reflection positivity of the Boltzmann functional $\omega_{H}$  
for a large class of reflection symmetric, globally gauge invariant Hamiltonians $H$, namely all those for which the matrix of coupling constants is positive semidefinite.  For such Hamiltonians $Z_{H}\geq1$.

We use this result to prove reflection positivity
for an even wider class of Hamiltonians,
namely those for which the matrix of coupling constants 
\emph{across the reflection plane}
is positive semidefinite. 

Neither result will require $H$ to be hermitian, but if this happens to be the case,
$Z_H$ is automatically nonzero, and
$\rho_{H}$ is the Gibbs state with respect to the Hamiltonian $H$.

\subsection{Hamiltonians}\label{Sect:Hamiltonians}
The class of Hamiltonians for which these reflection positivity results hold, is defined in terms of the matrix of coupling constants,
\be \label{eq:DefnCC}
	 J = (J_{\I \I'})\;,
	 \qsp{where} {\I,\I' \in \cP_{+}} \;.
\ee
By definition, these are the coefficients $J_{\I\I'}\in \C$ of the Hamiltonian $H$ in its expansion with respect to the basis $\B$,
	\be\label{eq:DefnH}
		H = - \sum_{{\phantom{,}\I, \I'\in \cP_{+}}  
		%\atop {\abs{\I}=\abs{\I'}}
		} J_{\I \I'} 
		\, \Theta(C_{\I})\cz   C_{\I'}\;. 
	%	\qsp{for} J_{\I\I'}\in \C\;.
	\ee
The following proposition expresses some relevant properties of $H$ in terms of the matrix $J$. Recall that $H$ is called reflection invariant if $\Theta(H) = H$, and globally gauge invariant 
if $\alpha(H)=H$, where $\alpha$ is the global gauge automorphism of \eqref{eq:GlobalGauge}. 

\begin{prop}\label{prop:propertiesH}
The Hamiltonian $H$ in \eqref{eq:DefnH} is 
\begin{itemize}
	\item[\bf RI:]
	reflection-invariant %, $\Theta(H) = H$, 
	if and only if 
	$J$ is hermitian, $J_{\I'\I} = \overline{J_{\I\I'}}$.
	\item[\bf GI:]
	globally gauge-invariant %, $\alpha(H) = H$,
	if and only if $J_{\I\I'} = 0$ for $\abs{\I} \neq \abs{\I'}$.
	\item[\bf H:]
	hermitian %, $H^{*} = H$, 
	if and only if 
	%$q_{\I}q_{\I'} J_{\I\I'}$
	$\zeta^{\frac{1}{2}k_{\I}(k_{\I}-1)+\frac{1}{2}k_{\I'}(k_{\I'}-1)}J_{\I\I'}$
	is real.
\end{itemize}
\end{prop}

\begin{proof}
The first statement is Proposition \ref{prop:reflectioninv}.
For the second statement, note that
the global gauge transformation $\alpha$ leaves the 
basis element
$\Theta(C_{\I}) \cz C_{\I'}$ fixed if
$\abs{\I} = \abs{\I'}$, and otherwise multiplies it by
$-1$. Linear independence of the basis $\B$ ensures that 
each term in the expansion of $H$ must be gauge invariant.
The third statement is a consequence of Proposition \ref{Prop:Adjoints}.
\end{proof}

\begin{prop}\label{Prop:HermitianForm}
If $H$ is reflection invariant, then the sesquilinear form $\lra{A,B}^{0}_{\hz}$ on 
$\A_{+}$ given by  \eqref{eq:Hsesq0} is hermitian, 
and $Z_{H}=\Tr(e^{-H})$ is real:
 	 \[
		 \Theta(H)=H
		 \quad\Rightarrow\quad
		 \lra{A,B}^{0}_{\hz} = \overline{\lra{B,A}^{0}}_{\hskip -5pt \hz}\;,
		 \qsp{and}
		 \overline{Z_{H}} = Z_{H}\;.
	 \]
If both $H$ is reflection invariant and  $Z_{H} \neq 0$, then the form $\lra{A,B}_{\hz}$ is defined in \eqref{eq:Hsesq} and is  hermitian.
\end{prop}

\begin{proof}
The operator $R$ of \S \ref{Sect:Reflection-InvariantFunctionals} equals $e^{-H}$ here.  So  $\Theta(R) = e^{-\Theta(H)}$, and if $H$ is reflection invariant, then 
so is $R$.
By the implication 
(b)$\Rightarrow$(d) of 
Proposition \ref{prop:Hermiticity}, 
the form $\lra{\,\cdot\,,\,\cdot\,}_{\rz  }$
is hermitian.  Also (b)$\Rightarrow$(a) ensures that $Z_{H} = \Tr(e^{-H}) = \Tr(\Theta(e^{-H}))= \overline{Z_{H}}$ is real. Hence if $Z_{H}\neq0$, the form $\lra{A,B}_{\hz}$ is also hermitian.
\end{proof}

\subsection{Reflection Positivity: Preliminary Results}
We now prove reflection positivity of the Boltzmann functional $\omega_{H}$
for Hamiltonians $H$ that arise from a positive semidefinite matrix $J$ of coupling constants.

\begin{thm}[\bf Reflection Positivity of $\omega_{H}$, Part I]
\label{Thm:ReflectionPositivityH-I}
Let $H \in \A$ be reflection symmetric and globally gauge invariant. 
If the matrix $J$ of coupling constants for $H$, defined in equation~\eqref{eq:DefnH},
is positive semidefinite, then $\omega_{H}$ is reflection positive on~$\A_{+}$.
\end{thm}

We give some preliminary results before proving the theorem.

\begin{lem}\label{Lemma:Key-1} 
Let $\I_{1}, \ldots, \I_{k}, \I'_{1}, \ldots \I'_{k} \in \cS_{+}$ and $\abs{\I_{j}} = \abs{\I_{j}'}$ for $j \geq 1$.  Then for all $\I_{0},\I'_{0} \in \cS_{+}$,
\be\label{eq:productoftraces}
		\overline{\Tr(C_{\I_{0}}\cdots C_{\I_{k}})}
		\Tr(C_{\I'_{0}}\cdots C_{\I'_{k}})
\ee
is nonzero only if $\abs{\I_{0}} = \abs{\I'_{0}}$.
\end{lem}

\begin{proof}
For every lattice point $i \in \Lambda$, let $k_{i}(\I)$ 
be $1$ if $i$ occurs in $\I = (i_1, \ldots, i_{s})$, and $0$ otherwise.
Then $s = k_{\I} =\sum_{i\in \Lambda} k_{i}(\I)$.
If $\Tr(C_{\I_{0}}\cdots C_{\I_{k}})$ is nonzero, 
then 
$\sum_{j=0}^{k} k_{i}(\I_{j})$ is even, as
every $i \in \Lambda$
must occur an even number of times.
Therefore, 
\[ 
\sum_{i\in \Lambda}\sum_{j=0}^{k} k_{i}(\I_{j})
= 
\sum_{j=0}^{k} \left(\sum_{i\in \Lambda}k_{i}(\I_{j})\right)
=
\sum_{j=0}^{k} k_{\I_{j}}
\]
is even. Since $\abs{\I} = k_{\I} \text{ mod }2$, the sum
$\sum_{j=0}^{k}\abs{\I_{j}}$ is even.

Similarly, one finds that $\sum_{j=0}^{k}\abs{\I'_{j}}$ is even
 if $\Tr(C_{\I'_{0}}\cdots C_{\I'_{k}})$ is nonzero.
Since $\abs{\I_{j}} = \abs{\I'_{j}}$ for $j\geq 1$ by assumption, 
we infer that $\abs{\I_{0}} = \abs{\I'_{0}}$
if \ref{eq:productoftraces} is nonzero.
\end{proof}

\begin{lem}\label{Lemma:Key-2}
Under the hypotheses of Lemma \ref{Lemma:Key-1},
\beq\label{eq:needslabel}
		& &\hskip -.5in \Tr
		\big((\Theta(C_{\I_{0}})\cz C_{\I'_{0}})
		\cdots (\Theta(C_{\I_{k}})\cz C_{\I'_{k}})\big)\nn
		& &\quad= 
		\overline{\Tr(C_{\I_{0}}\cdots C_{\I_{k}})}
		\Tr(
		C_{\I'_{0}}\cdots C_{\I'_{k}})\,.
\eeq
\end{lem}

\begin{proof}
Use the definition of $\cz$ to write 
\beq\label{eq:start}
	&&\hskip-.5in\Tr((\Theta(C_{\I_{0}})\cz C_{\I'_{0}})
		\cdots (\Theta(C_{\I_{k}})\cz C_{\I'_{k}}))\nn
		&&= \zeta^{\sum_{j=0}^{k} 
		 \abs{\I_{j}} \abs{\I'_{j}} }
		\Tr\big(\Theta(C_{\I_{0}})C_{\I'_{0}}
		\cdots \Theta(C_{\I_{k}})C_{\I'_{k}}\big)\,,
\eeq
and bring the terms of the form $\Theta(C_{\I_{j}})$ to the left.
In doing so, one has to exchange $\Theta(C_{\I_{j}})$ with $C_{\I'_{j'}}$
for each $j'< j$, yielding a factor 
\[(-1)^{\sum_{j'= 0}^{j}\abs{\I'_{j'}} \abs{\I_{j}}} = 
\zeta^{2\sum_{j'= 0}^{j}\abs{\I'_{j'}} \abs{\I_{j}}}\,.\]
The right hand side in equation \ref{eq:start} can thus be written
	\be\label{eq:littlemonster}
	\zeta^{\sum_{j=0}^{k} 
		 \abs{\I_{j}}\abs{\I'_{j}} + 2 \sum_{0\leq j'<j \leq k} \abs{\I_{j}}\abs{\I'_{j'}}}
		\Tr(\Theta(C_{\I_{0}}\cdots C_{\I_{k}})
		C_{\I'_{0}}\cdots C_{\I'_{k}})\,,
	\ee
where we used that $\Theta(C_{\I_{0}}) \cdots \Theta(C_{\I_{k}})$ equals
$\Theta(C_{\I_{0}} \cdots C_{\I_{k}})$. 

Using the factorization of the trace,
$\Tr(X_{-}X_{+}) = \Tr(X_{-})\Tr(X_{+})$
for $X_{\pm} \in \A_{\pm}$, and reflection invariance,
$\Tr(\Theta(X)) = \overline{\Tr(X)}$, given in Proposition~\ref{prop:propertiesTr}.d and e,
\eqref{eq:littlemonster} becomes
	\be\label{eq:littlemonster2}
	\zeta^{\sum_{j=0}^{k} 
		 \abs{\I_{j}}\abs{\I'_{j}} + 2 \sum_{0\leq j'<j \leq k} 
		 \abs{\I_{j}}\abs{\I'_{j'}}}
		\overline{\Tr(C_{\I_{0}}\cdots C_{\I_{k}})}
		\Tr(
		C_{\I'_{0}}\cdots C_{\I'_{k}})\,.	
\ee

Using Lemma \ref{Lemma:Key-1}, we rewrite the phase in \eqref{eq:littlemonster}
\be\label{eq:phaseis1}
	\zeta^{\sum_{j=0}^{k} 
	\abs{\I_{j}}\abs{\I'_{j}} + 2 \sum_{0\leq j'<j \leq k} 	
	\abs{\I_{j}}\abs{\I'_{j'}}}
	 =
	\zeta^{\big(\sum_{j=0}^{k}\abs{\I_{j}} \big)^{2}} = 1\,.
\ee
The last equality holds as $\sum_{j=0}^{k}\abs{\I_{j}}$ must be even, so its square is $0 \text{ mod }4$, and the phase vanishes.
Combining \eqref{eq:phaseis1} with \eqref{eq:littlemonster2}, the proof is complete.
\end{proof}

\begin{proof}[Proof of Theorem \ref{Thm:ReflectionPositivityH-I}]
Expand $A, B \in \A_{+}$ as 
\[
	A = \sum_{\I \in \cP_{+}} a_{\I} \, C_{\I}
\qsp{and} 
	B = \sum_{\I \in \cP_{+}} b_{\I}\,C_{\I}\;,
	\qsp{with} C_{\I}\in\B_{+}\,.
\]
We claim that the sesquilinear form 
$\lra{A,B}^{0}_{\hz} = \Tr(\Theta(A)\cz B\,e^{-H})$
can then be written in the form 
\beq\label{eq:niceexpansioninlem}
	\lra{A,B}^{0}_{\hz}
	&=&	\sum_{k=0}^{\infty}\frac{1}{k!}
	\sum_{\I_{0}, \ldots \I_{k}}
	\sum_{\I'_{0}, \ldots, \I'_{k}}
	\overline{a_{\I_{0}}}b_{\I'_{0}}
	J_{\I_{1},\I'_{1}} \cdots J_{\I_{k},\I'_{k}} 
	\label{eq:numerator2}\nn
	& & \qquad \times \;
\overline{\Tr(C_{\I_{0}}\cdots C_{\I_{k}})}\,
		\Tr(C_{\I'_{0}} \cdots C_{\I'_{k}})\,.
\eeq
From the power series 
for $e^{-H}$ with $H$ given by \eqref{eq:DefnH}, 
one obtains the expansion 
\beq\label{eq:exponentialexpansion}
\lra{A,B}^{0}_{\hz} &=&
	\sum_{k=0}^{\infty}\frac{1}{k!}
	\sum_{\I_{0}, \ldots \I_{k} \in \cP_{+}}
	\;
	\sum_{\I'_{0}, \ldots, \I'_{k}\in \cP_{+}}
	\overline{a_{\I_{0}}}\,b_{\I'_{0}}
	J_{\I_{1}\I'_{1}} \cdots J_{\I_{k}\I'_{k}} \label{eq:numerator}\nn
	& & \times \;
	\Tr((\Theta(C_{\I_{0}})\cz C_{\I'_{0}}) \,\cdots\, 
		(\Theta(C_{\I_{k}})\cz C_{\I'_{k}}))\,.
\eeq
The terms with $\I_{0}$ and $\I'_{0}$ arise from $A$ and $B$,
while the remaining $\I_{j}$, $\I'_{j}$ come from powers of 
$H$.
By Proposition \ref{prop:propertiesH},  global gauge invariance of $H$
ensures that $\abs{\I_{j}} = \abs{\I'_{j}}$ for all $j\geq 1$.
From Lemma \ref{Lemma:Key-2},
we conclude that 
\beq\label{eq:productagain}
		&&\hskip-.5in
		\Tr((\Theta(C_{\I_{0}})\cz C_{\I'_{0}}) \,\cdots\, 
		(\Theta(C_{\I_{k}})\cz C_{\I'_{k}}))
		\nn
		& &\quad =
		\overline{\Tr(C_{\I_{0}}\cdots C_{\I_{k}})}
		\Tr(
		C_{\I'_{0}}\cdots C_{\I'_{k}})\,.
\eeq
So by Lemma \ref{Lemma:Key-1}, $\abs{\I_{0}} = \abs{\I'_{0}}$
unless \eqref{eq:productagain} vanishes.
Using this and the expansion \ref{eq:exponentialexpansion}, 
one obtains \ref{eq:niceexpansioninlem}.

Let $\chi^{k}, \psi^{k}$ denote vectors with components 
\[
\chi^{k}_{\I_{1}, \ldots, \I_{k}}
=
\sum_{\I_{0} \in \cP_{+}} a_{\I_{0}} \Tr(C_{\I_{0}}\cdots C_{\I_{k}})\;,
\]
and
\[
\psi^{k}_{\I_{1}, \ldots, \I_{k}}
=
\sum_{\I_{0} \in \cP_{+}} b_{\I_{0}} \Tr(C_{\I_{0}}\cdots C_{\I_{k}})\;,
\] 
labelled
by $\cP_{+}^{k}$. Let 
$J^{\otimes k}_{\I_{1}, \ldots, \I_{k} ; \I'_{1},\ldots, \I'_{k}} := 
J_{\I_{1}\I'_{1}} \cdots J_{\I_{k} \I'_{k}}$ be the $k^{\rm th}$ 
tensor power of the matrix $J_{\I \I'}$.
Since $J_{\I \I'}$ is a positive semidefinite matrix, 
$J^{\otimes k}$ is also positive semidefinite.
Then 
	\be\label{eq:ValueRPExpectation}
		\lra{A,B}^{0}_{\hz}
		= \sum_{k=0}^{\infty}\frac{1}{k!}\,
		\langle{\chi^{k}, J^{\otimes {k}} \psi^{k} } \rangle\;,
	\ee
with the inner product 
\[
	\langle \chi^{k}, \psi^{k} \rangle := 
	\sum_{\I_{1}, \ldots \I_{k} \in \cP_{+}} 
	\overline{\chi^{k}_{\I_{1} \ldots \I_{k}}}\,
	\psi^{k}_{\I_{1} \ldots \I_{k}}\,.	
\]	

Setting $B=A$ one has $\psi^{k}=\chi^{k}$. 
Since each term in the sum \eqref{eq:ValueRPExpectation} is non-negative,
the theorem follows.
\end{proof}

\begin{cor}
If $A \in \A_{+}$ has the expansion \eqref{eq:basisexpansion}, then
under the conditions of Theorem \ref{Thm:ReflectionPositivityH-I}, one has
	\be\label{eq:est}
		\lra{A,A}^{0}_{\hz}
		\geq \abs{a_{\varnothing\varnothing}}^{2}\;.
	\ee  
\end{cor}

\begin{proof}
The right side of \ref{eq:est} is the $k=0$ term in \ref{eq:niceexpansioninlem}.
This yields a lower bound, as all the other terms are nonnegative by the proof of 
Theorem \ref{Thm:ReflectionPositivityH-I}. 
\end{proof}

\begin{prop}\label{Proposition:JPosSemiDef}
Suppose that the matrix $J$ of coupling constants for $H$, defined in \eqref{eq:DefnH}, 
is positive semidefinite. Then
$Z_{\beta H}$ is a non-decreasing function of $0 \leq \beta$ with $Z_{0} = 1$. %monotonically increasing in $\beta$.
In particular, $1 \leq Z_{\beta H}$ for all $0 \leq \beta$.
\end{prop}

\begin{proof}
Let $R=e^{-\beta H}$ and consider $Z_{\beta H}=\Tr(e^{-\beta H})=\Tr(R)$ for $\beta\geq0$.  Note that $Z_{0}=1$ by Proposition \ref{prop:propertiesTr}.a.  
Using Proposition \ref{Prop:Orthogonality} to evaluate the trace, one obtains 
\beq\label{DerivativeZ}
		\frac{dZ_{\beta H}}{d\beta}
		= -\Tr(He^{-\beta H})=-\Tr(HR)%\nn &=& 
 = \sum_{\I,\I'\in\cP_{+}} q_{\I}\,J_{\I\I'} \,r_{\I\I'}\,q_{\I'}\;,
	\eeq
with $q_{\I} = (-1)^{k_{\I}(k_{\I}-1)}$ as defined in equation \eqref{eq:q}.
Since the matrix $J_{\I\I'}$ is positive semidefinite, 
the Boltzmann functional $\omega_{H}$ is reflection positive by Theorem 
\ref{Proposition:JPosSemiDef}. 
The matrix $r_{\I\I'}$ of coefficients of $R=e^{-\beta H}$ is
positive semidefinite, as a consequence of Theorem \ref{Theorem:RPFunctionals}. 

It follows that the Hadamard product matrix $K$, with matrix elements $K_{\I\I'}=J_{\I\I'}\,r_{\I\I'}$, is also positive semidefinite.  From \ref{DerivativeZ}, we infer that 
	\be
		\frac{dZ_{\beta H}}{d\beta}
		=\sum_{\I,\I'\in\cP+} q_{\I}\, K_{\I\I'}\,q_{\I'}
		= \lra{q,Kq}_{\ell^{2}}\geq0\;.
	\ee
Thus $Z_{\beta H}$ is a non-decreasing function of $\beta$.
\end{proof}

\section{Necessary and Sufficient Conditions}\label{Sect:Which}
In Theorem \ref{Thm:ReflectionPositivityH-I} we have given sufficient conditions for reflection positivity of the Boltzmann functional $\omega_{H}$; 
it is reflection positive 
if $J\geq 0$, where $J$ is the matrix \eqref{eq:DefnCC} of couplings by which $H$
is defined. 

Now we establish a stronger result, providing necessary and sufficient conditions
in terms of the submatrix $J^0$ of $J$ that contains only the couplings 
between Majoranas on different sides of the reflection plane. If $J$ is positive 
semidefinite, then $J^0$ is positive semidefinite, but the converse does not hold.

In Section \ref{sec:RPomega} we prove that $\omega_{H}$ is reflection positive
if and only if $J^0 \geq 0$. 
Using this, 
we prove the analogous statement for the Gibbs functional $\rho_{H}$
in \S \ref{sec:RPgibbs}.

\subsection{Coupling Constants Across the Reflection Plane}
\label{Sect:CouplingAcross}
Let $H$ be reflection invariant, so that the coupling-constant matrix $J$ is hermitian.  Order the basis elements in $\B_{+}$ so $C_{\vn} = I$ is the first one, and consider the decomposition of $J$,
	\be\label{eq:matrix22}
		J 
		=	
		\bm{J_{\vn\vn}}{J_{\vn \I'}}{J_{\I\vn}}{J_{\I \I'}} 
		= \bm{E} {V^{*}} {V} {J^0}\;.
	\ee
Here $E = J_{\vn\vn}$ yields the additive constant $-E$ in $H$. Reflection invariance of $H$ ensures that $E$ is real.  

In fact $E$ is not of physical relevance.  It does not affect whether the functional $\omega_{H}$ is reflection positive. Furthermore it does not even enter the normalized Gibbs functional.   The energy shift $H \mapsto H - E $ multiplies both $\omega_{H}$ and $Z_{H}$ by $e^{E}$, so it does not affect their quotient $\rho_{H}$.   The column vector $V_{\I} = J_{\I\vn}$ has indices labelled by $\I \in \cP_{+} - \{\vn\} $, as does its hermitian adjoint $V^{*}$. The
hermitian matrix 
	\be  \label{Defn:CCreflection}
		J^{0} = (J^{0}_{\I,\I'})\;,
		\qsp{with indices} 
		{\I, \I' \in \cP_{+} - \{\vn\}}
	\ee 
is called the matrix of \emph{coupling constants 
across the reflection plane}.

The matrix decomposition \eqref{eq:matrix22}
corresponds to the four terms in the  decomposition
	\be\label{eq:DecompositionH}
		H 
		= H_{-} + H_{0} + H_{+} - E \,,
	\ee
where
	\be
		-H_{-} 
		= \sum_{\I \in \cP_{+} - \{\vn\}} 
		J_{\I\vn} \,\Theta(C_{\I})
		= \sum_{\I \in \cP_{+} - \{\vn\}} 
		V_{\I} \,\Theta(C_{\I}) \in \A_{-}
	\ee
is the sum of the interactions on one side of the reflection plane, namely on sites in $\Lambda_{-}$.  The reflection $H_{+}$ of $H_{-}$ is the interaction within $\Lambda_{+}$,
	\be
		-H_{+} 
		= \Theta(-H_{-})
		= \sum_{\I \in \cP_{+} - \{\vn\}} 
		\overline{V_{\I}} \,C_{\I}  \in \A_{+}\;.
	\ee
The interaction across the reflection plane is  
	\be
		-H_{0} 
		= \sum_{\I,\I' \in \cP_{+} - \{\vn\}} 
			J^{0}_{\I\I'} \,\Theta(C_{\I})\cz C_{\I'}\;.
	\ee

\subsection{Characterization of  Reflection Positivity}\label{sec:RPomega}
We give necessary and sufficient conditions on the
Hamiltonian $H\in \A$ for the Boltzmann functional 
\[\omega_{H}(A) = \Tr(Ae^{-H})\]
to be reflection 
positive on $\A_+$.

\begin{Remark}\label{Rk:GibbsRP}\rm
Reflection positivity
of $\omega_{H}$
means that the hermitian form on $\A_{+}$ defined by
\[
	\lra{A,B}^0_{\hz} = \Tr(\Theta(A)\cz B \cdot e^{-H})
\] 
is positive semidefinite; 
$0\leq\lra{A,A}^{0}_{\hz}$ for $A \in \A_{+}$. 
In particular,
\be \label{eq:Zgeq0}
	 Z_{H} = \Tr(e^{-H}) = \lra{I,I\,}^0_{\hz}  \geq 0\,.
\ee
If $Z_{H} \neq 0$, reflection positivity of 
the Boltzmann functional $\omega_{H}$ 
therefore implies 
reflection positivity of the (physically relevant) Gibbs functional $\rho_{H} = Z_{H}^{-1}\omega_{H}$. 
\end{Remark}

\begin{thm}[\bf Reflection Positivity of $\boldsymbol{\omega_{H}}$, Part II]\label{Thm:ReflectionPositivityH}
Let $H \in \A$ be reflection symmetric and globally gauge invariant. Let  $J^{0}$ be the matrix of coupling constants across the reflection plane, defined in \eqref{eq:matrix22}--\eqref{Defn:CCreflection}.  Then: 
\begin{itemize}
\item[(a)]
If $J^{0}$ is positive semidefinite, the functional $\omega_{H}$ is reflection positive on~$\A_{+}$. 
\item[(b)] Conversely, if there exists an $\varepsilon > 0$ such that 
$\omega_{\beta H}$ is reflection positive on $\A_{+}$  for all $\beta \in [0,\varepsilon)$, then the matrix $J^{0}$ is positive semidefinite.
\end{itemize}
\end{thm}

\begin{proof}
(a) Since $H$ is reflection invariant, we infer from Proposition \ref{prop:propertiesH} that $J$ is hermitian.  Writing $J$ as in \eqref{eq:matrix22}, recall that  reflection positivity of $\omega_{H}$ is independent of the value of $E$.  So for simplicity we can add a constant to $H$ so that $E=0$.  Now we approximate $J$ by  $J_{\vep}$ defined as the matrix
	\be
		J_{\vep} 
		:= \bm{0}{V^{*}}{V}{J^{0}} +\varepsilon \bm{0}{0}{0}{VV^{*}}\;,
	\ee	
where $0\leq\vep$ is a small parameter.  Here $VV^{*}$ denotes the matrix with elements $\lrp{VV^{*}}_{\I\I'}=V_{\I} \overline{V_{\I'}}$ with $\I,\I'\in \cP-\varnothing$.  Clearly $J_{\ep}\to J$ as $\vep\to0$, so that $H_{\vep}\to H$ as $\vep\to0$.   Hence $\omega_{H_{\ep}} \to \omega_{H}$ as $\vep\to0$.

Assume that the functional $\omega_{H_{\vep}}$ satisfies reflection positivity on $\A_{+}$ for every $\vep>0$.  Then the convergence explained above means that for $A\in\A_{+}$, the  expectations $\omega_{H_{\vep}}(\Theta(A)\circ A)\geq0$ converge to $\omega_{H}(\Theta(A)\circ A)\geq0$ as $\vep\to0$.  We infer that $\omega_{H}$ is reflection positive.

Now we show that $\omega_{H_{\vep}}$ does satisfy reflection positivity for every $\vep>0$.   In order to see this, we make a second modification to $J$, by adding the constant $\vep^{-1}$ to $H_{\ep}$.  Thus we obtain a new matrix of couplings $\widetilde {J}_{\vep}$ defined as  
	\be\label{eq:TildeJvep}
		\widetilde {J}_{\vep}
		= J_{\vep} + \frac{1}{\vep}\bm{1}{0}{0}{0}
		= \bm{0}{0}{0}{J^{0}} + \bm{\vep^{-1}}{V^{*}}{V}{\vep VV^{*}}\;.
	\ee 
The couplings $\widetilde {J}_{\vep}$ correspond to a Hamiltonian $\widetilde{H}_{\vep}$, that differs from $H_{\vep}$ only by an additive constant.  So $\omega_{\widetilde{H}_{\vep}}$ satisfies reflection positivity if and only if $\omega_{H_{\vep}}$ does.  Furthermore we can appeal to Theorem \ref{Thm:ReflectionPositivityH-I}, so it is sufficient to show that the matrix $\widetilde {J}_{\vep}$ is positive semidefinite for every $\vep>0$. 

We claim that $\widetilde{J}_{\vep}$ is positive semidefinite, since each of the two matrices on the right of \eqref{eq:TildeJvep} are positive semidefinite, as is the sum of two positive semidefinite matrices.   The first matrix on the right is positive semidefinite by the assumption that $J^{0}$ is positive semidefinite.  
The second matrix is also positive definite, since it can be written as
\[
 \bm{\vep^{-1}}{V^{*}}{V}{\vep VV^{*}} =  \bm{\vep^{-1/2}}{0}{\vep^{1/2}V}{0}  \bm{\vep^{-1/2}}{\vep^{1/2}V^{*}}{0}{0}\,.
\] 
 This  concludes the proof  that $\omega_{H}$ is reflection positive on $\A_{+}$.

(b).  Suppose that $\omega_{\beta H}$ is reflection positive on $\A_{+}$ for $\beta \in [0,\varepsilon)$.
Choose $A = \sum_{\I \in \cP_{+}} a_{\I} \,C_{\I}$ with $a_{\vn} = 0$, so $A$ is in the null space of the form $\lra{A,A}^{0}_{0,\Theta}$, as  
$\Tr(\Theta(A)\circ A) = \abs{a_{\vn}}^2 = 0$.
Reflection positivity then ensures that the first derivative cannot be negative,
\be\label{eq:der0beta}
	0 \leq \left.\frac{d}{d\beta} \lra{A,A}^{0}_{\beta H, \Theta}\right|_{\beta =0}
	= - \Tr((\Theta(A)\circ A)H)\,,
\ee
for otherwise reflection positivity would be violated for small $\beta$.
One can evaluate \eqref{eq:der0beta} in a fashion similar to the computation of \eqref{DerivativeZ}, but with $\Theta(A)\cz A$ replacing $R$.  

Expanding $\Theta(A)\circ A$ as 
$\sum_{\I,\I'} \overline{a_{\I}}a_{\I'} \Theta(C_{\I})\circ C_{\I'}$,
and using Proposition \ref{Prop:Orthogonality} to evaluate the trace, one obtains 
\be\label{eq:poswithF}
0\leq 
-\Tr(\Theta(A)\circ A H) = \sum_{\I,\I' \in \cP_{+} - \vn} \overline{(q_{\I}a_{\I})}\, J^{0}_{\I\I'} \, (q_{\I'}a_{\I'}),
\ee
with $q_{\I} = (-1)^{k_{\I}(k_{\I}-1)}$ as in \eqref{eq:q}.
As $a_{\vn}=0$, the sum restricts to $\cP_{+}{-}\vn$, and only $J^{0}$ contributes.
From equation \ref{eq:poswithF}, one then obtains
	\be \label{eq:J0pos}
		0 \leq 
		\lra{f, J^{0}f} \;.
	\ee
Since this holds for all $f \in \ell^2(\cP_{+})$ with $f_{\vn} = 0$, this assures that the matrix $J^{0}$ is positive semidefinite.  
\end{proof}

\begin{Remark}\rm
This theorem is somewhat similar in flavor to Schoenberg's Theorem \cite{S38a,S38b}, which states that $e^{-\beta K}$ is a positive 
definite kernel for all $\beta \in [0, \varepsilon)$ if and only if $K$ is conditionally negative definite.
In particular, it is striking that just as in Schoenberg's Theorem, $J$ is only required to be positive definite 
on a subspace of codimension 1. 
%(Cf. also \cite{S38a,S38b}.)
\end{Remark}

\subsection{Reflection Positive Gibbs Functionals}\label{sec:RPgibbs} 
Using Theorem~\ref{Thm:ReflectionPositivityH}, we obtain the following
necessary and sufficient conditions on 
$H$ for
the Gibbs functional 
 	\[
		\rho_{H}(A) =Z^{-1}_{H} \Tr(Ae^{-H})\,,
	\]
to be reflection positive on $\A_{+}$.
In this expression, $Z_{H} =\Tr(e^{-H})$ denotes the partition sum.

\begin{thm}[\bf Reflection Positivity of Gibbs Functionals]
\label{Theorem:RPGibbsFunctionals}
Let $H \in \A$ be a reflection symmetric, globally gauge invariant Hamiltonian.
\begin{itemize}
\item[(a)] Suppose that $Z_{H} \neq 0$, and that the matrix $J^{0}$
of coupling constants across the reflection plane 
is positive semidefinite.
Then $\rho_{H}$ is reflection positive, and $Z_{H} > 0$.
\item[(b)] If there exists an $\varepsilon >0$ such that 
$\rho_{\beta H}$ is reflection positive for all $\beta \in [0,\varepsilon)$, then
the matrix  $J^{0}$
of coupling constants across the reflection plane is positive semidefinite.
\end{itemize}
\end{thm}

\begin{Remark}\rm
If $H \in \A^{\rm even}$ is self-adjoint, then
the condition $Z_{\beta H} \neq 0$ is automatically satisfied for all $\beta \geq 0$.
\end{Remark}

\begin{proof}
(a) If $J^{0}$ is positive semidefinite, then $\omega_{H}$ is reflection positive 
by Theorem~\ref{Thm:ReflectionPositivityH}.
Reflection positivity of the Gibbs functional $\rho_{H}$ then follows by
Remark \ref{Rk:GibbsRP}.

(b)
The partition function $Z_{\beta H} = \Tr(e^{-\beta H})$ is analytic in $\beta$,  
and real by Corollary \ref{cor:Zisreal}. Since $Z_{0} = 1$, the 
expression
\[\rho_{\beta H}(X) = Z_{\beta H}^{-1} \Tr(Xe^{-\beta H})\]
is well defined and analytic in a neighborhood $U$ of $\beta = 0$. 
The inequality
$0 \leq \rho_{\beta H}(\Theta(A)\circ A)$ for $\beta \in U$ thus yields 
\[0 \leq  Z_{\beta H}\ \rho_{\beta H}(\Theta(A)\circ A) = \omega_{\beta H}(\Theta(A)\circ A)\,.\] 
Since this holds for all 
$A \in \A_{+}$ and $\beta \in U$, 
the Boltzmann functional $\omega_{\beta H}$ is reflection positive for all
$\beta \in U$, and $J^{0}$ is positive semidefinite by Theorem \ref{Thm:ReflectionPositivityH}.
\end{proof}

\section{Reflection Positivity for Spin Systems}\label{sec:RPSpinSyst}

From the corresponding result for Majoranas,
we now derive necessary and sufficient conditions for reflection positivity 
in the context of spin systems.
As in the case of Majoranas, these will be formulated in terms of the 
matrix of coupling constants across the reflection plane.

\subsection{Spin Algebras}
In spin models, the algebra of observables for a lattice site $j\in \Lambda$ is $M^{2}(\C)$,
spanned by $I$ and the Pauli spin matrices $\sigma_{j}^{1}, \sigma_{j}^{2}, \sigma_{j}^{3}$. The operators 
$\sigma_{j}^{a}$ and $\sigma_{j'}^{b}$ commute for $j\neq j'$, and otherwise satisfy 
the familiar relations
$
	\sigma_{j}^{a} \sigma_{j}^{b} = \delta^{ab}I + i\sum_{c}\epsilon_{abc}\sigma_{j}^{c}\,.
$
In this context, the full algebra of observables is 
\[\A^{\mathrm{spin}} = \bigotimes_{j\in \Lambda} M_{j}^2(\C)\,,\] 
and the algebra of observables on the $\pm$ side of the reflection plane
is
\[\A^{\mathrm{spin}}_{\pm} = \bigoplus_{j\in \Lambda_{\pm}} M_{j}^2(\C)\,.\] 
Define the operators $\Sigma_{(\I,A)}$ as the product of spins 
\[\Sigma_{(\I,A)} = \sigma_{i_1}^{a_1}\ldots\sigma^{a_k}_{i_k}\,.\]
They are labelled by sets of the form 
\be\label{eq:labeldef}
	(\I, A) := \{(i_1, a_1), \ldots, (i_{k},a_{k})\}\,,
\ee
where $i_{s}$ is a lattice point in $\Lambda$,  $a_{s}$ is a spin label in $\{1,2,3\}$,  
and $i_{s}\neq i_{t}$ for $s \neq t$.
Together with the identity $\Sigma_{\vn} := I$,
the operators $\Sigma_{(\I, A)}$ constitute an orthonormal basis of $\A$ 
with respect to the bilinear trace pairing,
\be\label{eq:TracepairS}
{\textstyle \Tr_{\rm spin}}(\Sigma_{(\I, A)} \Sigma_{(\I',A')}) = \delta_{AA'}\delta_{\I\I'}\,.
\ee

\begin{defini}[{\bf Standard Reflection}]\label{def:standardreflection}
The standard reflection $\Theta$ on  $\A^{\rm spin}$ 
is defined by $\Theta(\sigma^{a}_{j}) = - \sigma_{\vartheta(j)}^{a}$,
for $j \in \Lambda$ and $a \in \{1,2,3\}$. 
\end{defini}
The standard reflection satisfies
\be\label{eq:spinrefl}
\Theta(\Sigma_{(\I,A)}) = (-1)^{k_{\I}} \Sigma_{\vartheta(\I, A)}\,.
\ee

\subsection{Spin Hamiltonians}\label{sec:spin-models}

Any Hamiltonian $H^{\rm spin} \in \A^{\rm spin}$, not necessarily Hermitian, takes the form
\be\label{eq:spinhamiltonian}
H^{\rm spin} = - \sum_{k}\sum_{j_1,\ldots, j_{k}} \sum_{a_1, \ldots, a_{k}} J^{a_1}_{j_1}\ldots\,^{a_k}_{j_k} \,\sigma_{j_1}^{a_1}\ldots\sigma^{a_k}_{j_k}\,.
\ee
Partition $j_1 \ldots j_{k}$ into the sets $\vartheta(\I) \subseteq \Lambda_{-}$
and $\I' \subseteq \Lambda_{+}$, where both $\I$ and $\I'$ are subsets of $\Lambda_{+}$.
Using \eqref{eq:spinrefl} and setting 
\be \label{eq:coefssp}
J^{AA'}_{\vartheta(\I)\I'} = J^{a_1}_{j_1}\ldots{}^{a_k}_{j_k}\,,
\ee 
equation
\eqref{eq:spinhamiltonian} can be expressed as
\beq
H^{\rm spin} &=& 
-\sum_{{(\I,A)}\atop{(\I',A')}}
J^{AA'}_{\vartheta(\I)\I'}\, \Sigma_{(\vartheta(\I), A)}\Sigma_{(\I',A')} \label{eq:spinhamiltonian1}\\
&=&
-\sum_{{(\I,A)}\atop{(\I',A')}} (-1)^{k_{\I}} J^{AA'}_{\vartheta(\I)\I'}\, 
\Theta(\Sigma_{(\I, A)})\Sigma_{(\I',A')}\,.
\label{eq:spinhamiltonian2}
\eeq

\subsection{Mapping Spins to Majoranas}
Spin models map to Majorana models by a well-known transformation.  For a single site, this is similar to the infinitesimal rotation written in terms of Dirac matrices.  The tensor product construction, projected to a chiral subspace,  is known in the condensed matter literature as the Kitaev transformation.
This map $X \mapsto \widehat{X}$ from the algebra $\A^{\rm spin}$ of spins to the 
algebra $\A$ of Majoranas is constructed as follows.

Choose four Majoranas at site $j$ denoted $c^{\alpha}_{j}$, for $\alpha=1,2,3,4$.  
(The superscripts denote labels, not powers.) 
The Majoranas satisfy the Clifford relations
 $\{c^{\alpha}_{j}, c^{\beta}_{j'}\} = 2\delta^{\alpha\beta}\delta_{jj'}I$ and $c_{j}^{\alpha *} = c_{j}^{\alpha}$. 
They generate the Majorana algebra $\A$ indexed by
$\widehat{\Lambda} = \Lambda \times \{1,2,3,4\}$.

The product $\gamma_{j}^{5} = c_{j}^{1}c_{j}^{2}c_{j}^{3} c_{j}^{4}$ is both 
self adjoint and unitary, so $P^{5}_{j} = \half (I + \gamma_{j}^{5})$
is the projection corresponding to the $+1$ eigenvalue. 
The projections $P^{5}_{j}$ mutually commute, and also commute 
with all even elements of $\A$. 
Their product $P^{5} :=\prod_{j}P^{5}_{j}$ is called the \emph{chiral projection}.
It can be written as a product
\be
P^{5} = P^{5}_{-}P^{5}_{+}
\ee
of the two commuting projections 
$P^{5}_{\pm} = \prod_{j \in \Lambda_{\pm}} P^{5}_{j}$ in $\A_{\pm}$.

The map from spins to Majoranas
is given by
	\be\label{eq:Identification of Sigma}
		\widehat{\sigma}^{a}_{j}
		:= i c^{a}_{j} c^{4}_{j}
	\ee
on single spins $\sigma^{a}_{j}$, and extends to a linear map 
$\A^{\rm spin} \rightarrow \A$ by
\[
\widehat{\Sigma}_{(\I,A)} := 
\widehat{\sigma}^{a_{1}}_{j_{1}}\cdots \widehat{\sigma}^{a_{k}}_{j_{k}}\,.
\]
The resulting linear map $X \mapsto \widehat{X}$ 
is a homomorphism when restricted to $P^{5}$, in the sense that
for all $X, Y \in \A^{\rm spin}$, one has
\be\label{eq:almosthomo}
\widehat{XY}\,P_{5} = \widehat{X}\widehat{Y}\,P_{5}\,.
\ee

\subsection{Reflection Positivity for Spin Hamiltonians}
Recall that for a (not necessarily Hermitian) Hamiltonian $H \in H^{\rm spin}$,
the Boltzmann functional $\omega_{H}(X) = \Tr_{\rm spin}(Xe^{-H})$ is called refection positive on $\A_{+}$ if
\be\label{eq:forspins}
0 \leq
	\omega_{H}(\Theta(X) X)
	=  {\textstyle \Tr_{\rm spin}}(\Theta(X) X \,e^{-H})\,.
\ee
If the partition sum $Z_{H} = \Tr_{\rm spin}(e^{-H})$ is nonzero, then 
the Gibbs functional is defined by $\rho_{H}(X) := Z^{-1}_{H}\omega_{H}(X)$.
Reflection positivity of $\rho_{H}$ is equivalent to
\be\label{eq:forspins}
0 \leq
	\rho_{H}(\Theta(X) X)
	=  Z^{-1}_{H} {\textstyle \Tr_{\rm spin}}(\Theta(X) X \,e^{-H})\,.
\ee
From Theorem \ref{Thm:ReflectionPositivityH} for Majoranas, one derives the following 
characterization of reflection positivity for spin systems. It is given
in terms of the matrix $J^{0\, AA'}_{\I\I'}$ of coupling constants 
across the reflection plane. This is the submatrix of the matrix $J^{AA'}_{\I\I'}$
of coupling constants \eqref{eq:coefssp} with $\I \neq \vn$ and $\I' \neq \vn$.

\begin{thm}[\bf Reflection Positivity for Spins]\label{thm:maincorollarySpinSyst}
Let $H \in \A^{\mathrm{spin}}$ be a (not necessarily Hermitian) 
reflection invariant Hamiltonian.
\begin{itemize}
\item[(a)]  If the matrix 
$i^{k_{\I} + k_{\I'}} J^{0\,AA'}_{\,\vartheta(\I)\I'}$  
is positive semidefinite, then 
the Boltzmann functional $\omega_{H}$ is reflection positive.
If $Z_{H}\neq 0$, then $Z_{H}> 0$, and the Gibbs state $\rho_{H}$ is reflection positive. 
\item[(b)]
If there exists an $\varepsilon>0$  such that either $\omega_{\beta H}$ or 
$\rho_{\beta H}$ is reflection positive on  $\A^{\mathrm{spin}}_{+}$ 
for all $\beta\in[0,\varepsilon)$, then the matrix 
$i^{k_{\I} + k_{\I'}} J^{0\,AA'}_{\,\vartheta(\I)\I'}$ is positive semidefinite.
\end{itemize}
\end{thm}

\begin{Remark}\rm
The requirement that $H$ is reflection invariant is equivalent to Hermiticity of
the matrix
$i^{k_{\I} + k'_{\I}} J^{AA'}_{\vartheta(\I)\I'}$.
Furthermore, the requirement that $Z_{H}\neq 0$ is automatically fulfilled
if $H$ is Hermitian. 
\end{Remark}

\begin{proof}
It suffices to prove (a) and (b) for the Boltzmann functional $\omega_{H}$. 
Statement (a) for the Gibbs functional $\rho_{H}$ then follows from 
Remark \ref{Rk:GibbsRP}. Following 
word by word 
the proof of Theorem \ref{Theorem:RPGibbsFunctionals}.b,
one obtains statement (b) for $\rho_{H}$ 
from statement (b) for $\omega_{H}$.

(a): 
Since $\Theta(c^{a}_{\vartheta(j)}c^{4}_{\vartheta(j)}) = c^{a}_{j}c^{4}_{j}$,
the 
Hamiltonian $H^{\rm spin}\in \A$ of equation \ref{eq:spinhamiltonian}
with coefficients \ref{eq:coefssp} gives rise
to the Hamiltonian 
\be\label{eq:Hhat}
\widehat{H} = - \sum_{{(\I,A)}\atop{(\I',A')}} J{}^{AA'}_{\I\,\I'} i^{k + k'}
\Theta\Big(c^{a_{1}}_{\vartheta(i_{1})} c^{4}_{\vartheta(i_{1})} 
\ldots c^{a_k}_{\vartheta(i_k)} c^{4}_{\vartheta(i_k)}\Big)\;
c^{a'_{1}}_{i'_{1}} c^{4}_{i'_{1}} \ldots c^{a'_k}_{i'_{k'}} c^{4}_{i'_{k'}}
\ee
in the Majorana algebra $\A$.
Equation \eqref{eq:Hhat} can thus be written 
\[
\widehat{H}= - \sum_{\widehat{\I}, \widehat{\I}'} J^{\rm M}_{\widehat{\I} \,\widehat{\I}'}
\Theta(C_{\widehat{\I}}) \circ C_{\widehat{\I}'},
\]
where $J^{\rm M}$ is the matrix of Majorana coupling constants. 
It equals 
\be\label{eq:needsnumber}
J^{\rm M}_{\widehat{\I} \,\widehat{\I}'} = 
i^{k_{\I}+k_{\I'}} J{}^{AA'}_{\vartheta(\I)\I'}
\ee
for the indices
\beq
\widehat{\I} &=& ((i_1,a_1), (i_1,4), \ldots, (i_{k}, a_{k}), (i_{k}, 4))\,,\\\label{eq:indx}
\widehat{\I}' &=& ((i'_1,a'_1), (i'_1,4), \ldots, (i'_{k}, a'_{k'}), (i'_{k'}, 4))\nonumber
\eeq
and zero elsewhere.
With respect to an appropriate choice of basis, 
the matrix $i^{k_{\I}+k_{\I'}}J{}^{AA'}_{\vartheta(\I)\I'}$ 
is the only nonzero block in 
$J^{\rm M}_{\widehat{\I} \,\widehat{\I}'}$. Therefore, the 
latter is  
positive semidefinite if and only if 
the former is.
The same holds for the matrices
$J^{M0}_{\widehat{\I} \,\widehat{\I}'}$ and 
$i^{k_{\I}+k_{\I'}}J^{0\,AA'}_{\,\vartheta(\I)\I'}$
of couplings across the reflection plane.

The Majorana Hamiltonian $\widehat{H}$ is 
globally gauge invariant since each spin involves two Majoranas, 
and it is reflection invariant as $\widehat{J}$ is Hermitian.
Since $J^{M0}_{\widehat{\I} \,\widehat{\I}'}$ is positive semidefinite, 
Theorem \ref{Theorem:RPGibbsFunctionals}
yields reflection positivity of $\widehat{H}$.
This implies
reflection positivity of $H$, since
\beqs
\Tr{}_{\mathrm{spin}}(\Theta(X)X e^{-H}) &=& \Tr{}_{M} (\Theta(\widehat{X})\widehat{X} e^{-\widehat{H}} P^{5})\\
& = & 
\Tr{}_{M} (\Theta(\widehat{X}P^{5}_{+})(\widehat{X}P^{5}_{+}) e^{-\widehat{H}} )
\geq 0\,.
\eeqs
Here, we used $\Tr_{\rm spin}(X) = \Tr_{M}(\widehat{X} P^5)$, equation
\eqref{eq:almosthomo}, and the fact that $P^{5}_{+}$ and 
$\Theta(P^{5}_{+}) = P^{5}_{-}$ commute 
with the other factors, with $P^{5} = P^{5}_{-}P^{5}_{+}$.

(b):
This is analogous to the proof of Theorem \ref{Thm:ReflectionPositivityH}.b.
Choose $X \in \A^{\rm spin}_{+}$ such that $\Tr_{\rm spin}(\Theta(X)X) = 0$.
Expand $X$ as 
\[
	X = \sum_{(\I',A')} x^{A'}_{\I'}\Sigma_{(\I',A')}\,, 
\] 
with 
the coefficient $b_{\vn}$ of $\Sigma_{\vn} = I$ equal to zero.
Using equation \eqref{eq:spinrefl}, one finds
\[
\Theta(X) = \sum_{(\I, A)} (-1)^{k_{\I}}\overline{x}{}^{A}_{\I} \Sigma_{\vartheta(\I, A)}\,.
\]
Since $\rho_{\beta H}(\Theta(X)X) = \Tr_{\rm spin}(\Theta(X)X e^{-\beta H})$ is nonnegative 
and zero for $\beta = 0$, one finds
\be
0 \leq \left.\frac{d}{d\beta}  {\textstyle \Tr_{\rm spin}}(\Theta(X)X e^{-\beta H}) \right|_{\beta = 0}
= -{\textstyle \Tr_{\rm spin}}(\Theta(X)X H)\,.
\ee
Using the expansion \eqref{eq:spinhamiltonian1}
and the orthogonality relations \eqref{eq:TracepairS}
of $\Sigma_{(\I, A)}$ with respect to 
the trace pairing, 
one thus obtains 
\[
0 \leq \sum_{{(\I,A)}\atop{(\I',A')}} (-1)^{k_{\I}} \overline{x}{}^{A}_{\I} J^{AA'}_{\vartheta(\I)\I'} x^{A'}_{\I'}\,.
\]
Since $x_{\vn} = 0$, only the coupling constants $J^{0\,AA'}_{\,\vartheta(\I)\I'}$
across the reflection plane contribute.
Substituting $y^{A}_{\I} := i^{k_{\I}}x^{A}_{\I}$ yields
\[
0 \leq \sum_{{(\I,A)}\atop{(\I',A')}}  \overline{y}{}^{A}_{\I} 
\left( i^{k_{\I} + k_{\I'}}J^{0\,AA'}_{\,\vartheta(\I)\I'}\right)
 y^{A'}_{\I'}\,,
\]
so that $i^{k_{\I} + k_{\I'}}J^{0\,AA'}_{\,\vartheta(\I)\I'}$ is positive semidefinite, as required.
\end{proof}

\section{Automorphisms that Yield New Reflections}\label{sec:reflectionsandgauge}
In Sections \ref{Sect:Which} and \ref{sec:RPSpinSyst}, we have given a characterization of 
reflection positivity with respect to a standard reflection $\Theta$. 
In this section, we show how these results extend to other reflections 
$\Theta' = \alpha^{-1} \Theta \alpha$, where $\alpha$ is an automorphism.
The special case where $\alpha$ is a gauge transformation, can be very useful in applications. 

\subsection{Relation to Other Reflections}
We formulate this in the more general context of a $\mathbb{Z}_{2}$-graded algebra $\A$ 
which is the super tensor product of two isomorphic subalgebras
$\A_{+}$ and $\A_{-}$.
This means that $\A$ is $\A_{+}\otimes\A_{-}$ as a vector space, 
with multiplication defined by  
\[
(A \otimes B)(A'\otimes B') = (-1)^{\abs{A'}\abs{B}} AA'\otimes BB'
\]
on homogeneous elements.
The twisted product $A \circ B$ is then defined 
as in Definition \ref{def:twistprod}. It reduces to the ordinary product 
on algebras that are purely even, such as the spin algebra $\A^{\rm spin}$.

A reflection $\Theta \colon \A \rightarrow \A$ is  an antilinear automorphism such that  $\Theta(\A_{\pm})=\A_{\mp}$ and $\Theta^{2} = I$. 
Two different reflections $\Theta$ and $\Theta'$ are related by the linear automorphism 
$\beta := \Theta \Theta'$, which maps $\A_{\pm}$ to $\A_{\pm}$, and satisfies 
$\Theta \beta = \beta^{-1}\Theta$.  
Conversely, if $\Theta$ is a reflection and $\beta$ satisfies 
$\beta(\A_{\pm}) = \A_{\pm}$ and 
$\Theta \beta = \beta^{-1}\Theta$, then 
$\Theta' := \Theta\beta$ is also a reflection.

Recall that a linear functional $\omega \colon \A \rightarrow \C$ is reflection positive on $\A_{+}$ with respect to $\Theta'$,  if $0 \leq \omega(\Theta'(A)\circ A)$ for all $A \in \A_{+}$.
If $\Theta'$ is related to $\Theta$ by a square $\beta = \alpha^{2}$,
then reflection positivity with respect to $\Theta$ and $\Theta'$ are related as follows.

\begin{prop}
Let $\alpha$ be a linear automorphism of $\A$ such that 
$\alpha(\A_{\pm}) = \A_{\pm}$ and $\Theta \alpha = \alpha^{-1} \Theta$.
Let
\[
\Theta' :=
\alpha^{-1} \Theta \alpha 
\,.\] 
Then the pullback
$\alpha^{-1 *} \omega(A) := \omega(\alpha^{-1}(A))$
is reflection positive with respect to 
$\Theta$
on $\A_{+}$, if and only if $\omega$ is reflection positive with respect to
$\Theta'$ on $\A_{+}$.
\end{prop}

\begin{proof}
Since $\alpha$ is a linear automorphism,
 $\alpha(A_{-} \circ A_{+}) = \alpha(A_{-}) \circ 
 \alpha(A_{+})$ for $A_{\pm} \in \A_{\pm}$.
 %$A_{-} \circ A_{+}$.
For $A \in \A_{+}$, one has
\beqs
\alpha^{-1 *}\omega(\Theta(A)\circ A) &=& 
\omega(\alpha^{-1} \Theta(A)\circ \alpha^{-1}(A))\\ &=&
\omega(\Theta' (\alpha^{-1}(A)) \circ \alpha^{-1}(A)))\,.
\eeqs
Thus the first term is positive for all $A\in \A_{+}$, if and only if the last term is positive.
\end{proof}

We apply this to the algebras 
of Majoranas and spins, with
the Gibbs functional
$\rho_{H}(A) = Z^{-1}_{H}\Tr(Ae^{-H})$. 

\begin{cor}\label{cor:trafo}
The Hamiltonian $H' := \alpha(H)$ is invariant under
the reflection $\Theta$ if and only if $H$ is invariant under $\Theta':= \alpha^{-1}\Theta \alpha$.
The Gibbs functional $\rho_{H}$ is reflection positive 
with respect to $\Theta'$ on $\A_{+}$, if and only if 
$\rho_{H'}$ is reflection positive 
with respect to $\Theta$ on $\A_{+}$.
\end{cor}
\begin{proof}
The first statement follows 
as $\Theta(\alpha (H)) = \alpha(H)$ is equivalent to 
$\alpha^{-1}\Theta \alpha (H) = H$.
For the second statement, note that
the normalized trace is unique 
on the algebras of Majoranas and spins. Thus
 $\alpha^{*}\Tr = \Tr$ for every automorphism $\alpha$, and one has
\beqs
\alpha^{-1 *}\rho_{H}(A) 
&=& Z^{-1}_{H}\Tr(\alpha^{-1}(A) e^{-H}) 
= 
Z^{-1}_{H}\Tr(\alpha(\alpha^{-1}(A) e^{-H}))\\
&=& Z^{-1}_{H} \Tr(A e^{-\alpha(H)})
= \rho_{\alpha(H)}(A)
\,.
\eeqs
\end{proof}

Note that in the above, we do \emph{not} require $\Theta$, $\Theta'$ or $\alpha$ to respect 
the involution $*$ on the algebra $\A$.
If $\A$ is either the spin algebra or the algebra of Majoranas, then 
the canonical reflection $\Theta$ preserves the involution.
In this case, $\Theta' = \alpha^{-1} \Theta \alpha$ will
preserve the involution if and only if $\alpha^2$ does so.

\subsection{Gauge Automorphisms}

In the context of a (super) tensor product $\A$
of $\Z_{2}$-graded $*$-algebras $\A_{j}$
\[
	\A = \bigotimes_{j \in \Lambda} \A_{j},
\]
we define the \emph{gauge automorphism} $\alpha_{\tau}$, 
parameterized by 
a collection $\{\tau_{j}\}_{j \in \Lambda}$ of automorphisms of $\A_{j}$, as
\[
	\alpha_{\tau} := \otimes_{j \in \Lambda} \tau_{j}\,.
\]
If the $\A_{j}$ can be canonically identified which each other, 
and all $\tau_{j}$ are the same, then
$\alpha_{\tau}$ is called a \emph{global} gauge transformation.

Suppose that $\A$ has a reflection $\Theta$ such that 
$\Theta(\A_{j})$ is isomorphic to $\A_{\vartheta(j)}$.
Then the gauge automorphism $\alpha_{\tau}$ is called 
\emph{reflection invariant} if
$\tau_{j} = \Theta \tau^{-1}_{\vartheta(j)} \Theta$ for all $j \in \Lambda$.
Every reflection invariant gauge automorphism 
satisfies 
\[
\alpha_{\tau}(\A_{\pm}) = \A_{\pm} \qsp{and} 
\alpha_{\tau} \Theta = \Theta \alpha^{-1}_{\tau}\,.
\]

\subsubsection{Majorana Algebras with 1 generator}
In the case of the Majorana algebra generated by $c_{j}$ with $j \in \Lambda$,
$\A_{j}$ is the two-dimensional algebra generated by $I$ and $c_{j}$, 
and the only two automorphisms are $\tau_{j}(c_{j}) = \pm c_{j}$.
There is a unique nontrivial global gauge automorphism $c_{j}\mapsto -c_{j}$.

\subsubsection{Majorana Algebras with 4 generators}
In the case of the algebra generated by Majoranas $c^{\alpha}_{j}$
with $j \in \Lambda$ and $\alpha \in \{1,2,3,4\}$, the algebra $\A_{i}$ 
is the Clifford algebra 
$\mathrm{Cl}(4,\C)$ generated by the $c^{\alpha}_{i}$ with $i$ fixed.
The automorphisms $\tau_{j}$ can be taken to be conjugation by an invertible
element $g_{j} \in \mathrm{Cl}(4,\C)^{\times}$, that is,  
$\tau_{j}(A) = g_{j}A g^{-1}_{j}$. 
The spin group $\mathrm{Spin}(4)$ is the group of
even elements 
$g \in \mathrm{Cl}(4,\C)^{\times}$ such that
$g c^{\alpha} g^{-1} = R^{\alpha}_{\beta}c^{\beta}$
for some $R \in \mathrm{SO}(4,\R)$. 

\subsubsection{Spin Algebras}\label{sec:Gauge Automorphism}
In the next section the most relevant case will be 
the spin algebra $\A^{\rm spin}$, where
$\A_{i}$ is the purely even algebra $M^2(\C)$.
If $\tau_{i}$ is conjugation by a matrix $g_{i}\in \mathrm{SL}(2, \C)$,
we denote the gauge automorphism 
corresponding to the collection $\{g_{j}\}_{j \in \Lambda}$
by $\alpha_{g}$.
The requirement $g_{\vartheta(j)} = \Theta g^{-1}_{j}\Theta$
translates to $g_{\vartheta(j)} = g_{j}^{*}$.
It is an automorphism 
of $*$-algebras if and only if $g_{i}\in \mathrm{SU}(2,\C)$ for every 
$i\in \Lambda_{+}$.

\section{Examples of Spin Models}\label{sec:exspin}
We apply the characterization 
of reflection positivity in Theorem \ref{thm:maincorollarySpinSyst}
to a number of spin systems: the Ising model, 
the quantum rotator, and the anti-ferromagnetic Heisenberg model.  
Nearest neighbor couplings are treated in \S \ref{sec:nn}, and long range interactions
in \S \ref{sec:LR}.

Many of these examples are well-understood, and we include those mainly to show that they have a natural interpretation within our general framework.  Some relevant references are \cite{DysonLiebSimon1976,FroehlichIsraelLiebSimon1978,DysonLiebSimon1978,Froehlich-Lieb1978,Biskup}.  

In this section, the lattice $\Lambda$ has 
a geometric interpretation. It is a finite, fixed point free subset of  a manifold $\calm$ 
with involution $\vartheta_{\calm}$, as explained in \S \ref{sec:latsec}. 
An important example is $\calm = \R^{d}$ with $\vartheta \colon \R^{d} \rightarrow \R^{d}$
the orthogonal reflection in a hyperplane $\Pi$.
Periodic boundary conditions can be handled by taking
$\calm = \mathbb{T}^{d}$ the $d$-dimensional torus.

\subsection{Nearest Neighbor Couplings}\label{sec:nn}
The nearest neighbor \emph{Heisenberg model}
is given in terms of the Pauli matrices $\sigma^{a}_{j}$ on a lattice $j \in \Lambda$ 
by the Hamiltonian
	\be\label{eq:heismodel}
		-H = \sum_{a=1}^{3}\sum_{\lra{jj'}} J^{a}_{jj'} 
	\sigma^{a}_{j}\sigma^{a}_{j'}+ 
	\sum_{a=1}^{3}\sum_{j}h^{a}_{j}\sigma^{a}_{j}\;.
	\ee
Here the sum is over the nearest neighbor pairs $\lra{jj'}$, and $J^{a}_{jj'} = J^{a}_{j'j}$.
As $H$ is Hermitian, the partition sum $Z_{H} = \Tr(e^{-\beta H})$ is nonzero.

In order to define nearest neighbor models, 
we assume that the lattice $\Lambda \subseteq \calm$ has the property that 
``bonds are perpendicular to the reflection hyperplane''. 
This means that two lattice points $j \in \Lambda_+$ and $j' \in \Lambda_-$ 
can only be nearest neighbors if $j' = \vartheta(j)$. 
(For example, this is the case in Fig.\ \ref{fig:cubic} and Fig.\ \ref{fig:honey}.)  

Let $J^{0\, AA'}_{\, \vartheta(\I)\I'}$ denote
the matrix of couplings across the reflection plane, defined in
\eqref{eq:spinhamiltonian}, \eqref{eq:coefssp}.
It is given by 
\[
J^{0\, AA'}_{\, \vartheta(\I)\I'} = J^{a}_{\vartheta(j)j'}
\]
for the indices 
$(\I,A) = \{(j,a)\}$ and $(\I',A') = \{(j',a)\}$ 
of equation \eqref{eq:labeldef}, and
zero in all other components.
Here
$j, j'\in \Lambda_{+}$ and $a \in \{1,2,3\}$. 
Note that $J^{a}_{\vartheta(j)j'}$ is only nonzero if $j=j'$, 
as sites $j' \in \Lambda_{+}$ and $\vartheta(j) \in \Lambda_{-}$ 
on different sides of the reflection plane 
can only be neighbors if $j= j'$.

\subsubsection{Anti-Ferromagnetic Heisenberg Models}
In order to show reflection positivity for the anti-ferromagnetic Heisenberg model, 
we restrict the coupling constants in \eqref{eq:heismodel} as follows:

The full matrix of coupling constants 
is $\vartheta$-symmetric,
$J^{a}_{jj'} = J^{a}_{\vartheta(j)\vartheta(j')}$.
The external field is
antisymmetric, 
$h^{a}_{\vartheta(j)} = - h^{a}_{j}$,
and couplings across the reflection plane
are anti-ferromagnetic, $J_{\vartheta(j)j}\leq 0$.

\begin{prop}[{\bf Anti-ferromagnetic Heisenberg Model}]
For the above restrictions on the coupling constants in 
the Hamiltonian $H$ of \eqref{eq:heismodel},
the Gibbs state $\rho_{\beta H}$ is reflection positive
with respect to the standard reflection $\Theta(\sigma^{a}_{j}) = -\sigma^{a}_{\vartheta(j)}$.
\end{prop} 
\begin{proof}
Under the standard reflection $\Theta$, the first term on the right side of 
\eqref{eq:heismodel} is invariant if $J^{a}_{jj'} = J^{a}_{\vartheta(j)\vartheta(j')}$, 
while the second term is
invariant if the external field satisfies $h^{a}_{\vartheta(j)} = - h^{a}_{j}$.  
By Theorem \ref{thm:maincorollarySpinSyst}, 
the Gibbs state $\rho_{\beta H}$ is
reflection positive for all $\beta \geq 0$,   
if and only if the matrix
$i^{k_{\I} + k_{\I'}}J^{0\, AA'}_{\, \vartheta(\I)\I'}$
is positive semidefinite.
As $k_{\I} = k_{\I'} = 1$, this matrix is diagonal with 
entries $-J^{a}_{\vartheta(j)j}$, labelled by the 
$j \in \Lambda_{+}$ for which $\vartheta(j) \in \Lambda_{-}$.
This matrix is
positive definite if and only if $J^{a}_{\vartheta(j)j} \leq 0$.
\end{proof}

This includes the usual anti-ferromagnetic Heisenberg model, with constant couplings 
$J^{1}_{ij} = J^{2}_{ij} = J^{3}_{ij} = J \leq 0$, and vanishing external field $h^{a}_{ij}= 0$.
The \emph{quantum rotator model} is the special case 
$J^{3}_{ij} = 0$, and 
the \emph{Ising model} is the special case
$J^{2}_{ij} = J^{3}_{ij} = 0$. 
By the above proposition, they are reflection positive 
in the anti-ferromagnetic case of negative coupling constants
with vanishing external field $h^{a}_{j}$.

\subsubsection{Ferromagnetic Quantum Rotator Model}
The next example illustrates the gauge transformation method introduced in 
\S \ref{sec:reflectionsandgauge}.
In order to show reflection positivity for the ferromagnetic quantum rotator model, 
we restrict the coupling constants
in \eqref{eq:heismodel} as follows:

We require $J^{3}_{jj'} = 0$ and $0 \leq J^{a}_{j'j}$ for $a = 1,2$.
(In fact, the proof only uses that the bonds $j'= \vartheta(j)$
across the reflection plane 
are ferromagnetic.)
We assume that the couplings are symmetric around the reflection plane,
$J^{a}_{jj'} = J^{a}_{\vartheta(j)\vartheta(j')}\leq 0$ for $a = 1,2$.
Finally,  
we require that the first two components of the
external field are reflection symmetric, 
$h^{a}_{\vartheta(j)} = h^{a}_{j}$ for $a = 1,2$,
and that the third component is antisymmetric, 
$h^{3}_{\vartheta(j)} = -h^{3}_{j}$. 

\begin{prop}[\bf Ferromagnetic Quantum Rotator]
With the above restrictions on the coupling constants in 
the Hamiltonian $H$ of \eqref{eq:heismodel},
the Gibbs state $\rho_{\beta H}$ is reflection positive
with respect to the anti-linear reflection
$\Theta'$ that satisfies 
\beq\label{eq:ferroreflection}
	\Theta'(\sigma^{1}_{j}) = \sigma^{1}_{\vartheta(j)}, \quad
	\Theta'(\sigma^{2}_{j}) = \sigma^{2}_{\vartheta(j)}, \qsp{and}
	\Theta'(\sigma^{3}_{j}) = -\sigma^{3}_{\vartheta(j)}.
\eeq
\end{prop}

\begin{proof}
We use the gauge transformation $\alpha_{g}$
of \S \ref{sec:Gauge Automorphism}, with $g_{j} = e^{i\frac{\pi}{4}\sigma_{j}^{3}}$
for $j \in \Lambda_{+}$ and 
$g_{j} = e^{-i\frac{\pi}{4}\sigma_{j}^{3}}$
for $j \in \Lambda_{-}$.
This yields the clockwise rotation over $\pi/2$ 
around the third axis, 
\be\label{eq:gaugeno1}
	\alpha_{g}(\sigma^{1}_{j}) = - \sigma^{2}_{j},\;
	\alpha_{g}(\sigma^{2}_{j}) = \sigma^{1}_{j},\;
        \alpha_{g}(\sigma^{3}_{j}) = \sigma^{3}_{j}
\qsp{for} j \in \Lambda_{+}\,,
\ee
and the counterclockwise rotation
\be\label{eq:gaugeno2}
	\alpha_{g}(\sigma^{1}_{j}) = \sigma^{2}_{j},\;
	\alpha_{g}(\sigma^{2}_{j}) = -\sigma^{1}_{j},\;
        \alpha_{g}(\sigma^{3}_{j}) = \sigma^{3}_{j}
\qsp{for} j \in \Lambda_{-}\,.
\ee
After the gauge transformation, 
the Hamiltonian
$H$ of 
\eqref{eq:heismodel} becomes
$H'= \alpha_{g}(H)$, 
which decomposes as
$H'= H'_{+} + H'_{0} + H'_{-}$. Here
\[
-H'_{+}  =   
        \sum_{\lra{jj'}} J^{1}_{jj'} 
	\sigma^{2}_{j}\sigma^{2}_{j'}+
	\sum_{\lra{jj'}} J^{2}_{jj'} 
	\sigma^{1}_{j}\sigma^{1}_{j'}\\
	 +
	\sum_{j} h^{2}_{j}\sigma^{1}_{j} - h^{1}_{j}\sigma^{2}_{j} + h^{3}_{j}\sigma^{3}_{j},
\]
with the sum over nearest neighbors $j, j' \in \Lambda_{+}$.
Similarly, \[
-H'_{-}  =   
        \sum_{\lra{jj'}} J^{1}_{jj'} 
	\sigma^{2}_{j}\sigma^{2}_{j'}+
	\sum_{\lra{jj'}} J^{2}_{jj'} 
	\sigma^{1}_{j}\sigma^{1}_{j'}
	+
	\sum_{j} -h^{2}_{j}\sigma^{1}_{j} + h^{1}_{j}\sigma^{2}_{j} + 
	h^{3}_{j}\sigma^{3}_{j},
\]
with the sum over nearest neighbors $j, j' \in \Lambda_{-}$. Finally,  
\[
-H'_{0} = 
        \sum_{j} -J^{1}_{\vartheta(j)j} 
	\sigma^{2}_{\vartheta(j)}\sigma^{2}_{j}+
	\sum_{j} -J^{2}_{\vartheta(j)j} 
	\sigma^{1}_{\vartheta(j)}\sigma^{1}_{j}\,,
\]
where $j \in \Lambda_{+}$ has $j'\in \Lambda_{-}$ as a nearest neighbor.

The Hamiltonian $H'$ is invariant under
the standard reflection defined by $\Theta(\sigma^{a}_{j}) = -\sigma^{a}_{\vartheta(j)}$,
as long as
$J^{a}_{\vartheta(j)\vartheta(j')} = J^{a}_{jj'}$ and 
$h^{1}_{\vartheta(j)} = h^{1}_{j}$, $h^{2}_{\vartheta(j)} = h^{2}_{j}$, 
and $h^{3}_{\vartheta(j)} = - h^{3}_{j}$.
The matrix of coupling constants across the reflection plane is positive semidefinite 
if $0 \leq J^{1}_{\vartheta(j)j}$ and $0 \leq J^{2}_{\vartheta(j)j}$.
From Theorem \ref{thm:maincorollarySpinSyst}, we see that under these conditions, 
the Gibbs state $\rho_{\beta H'}$ for the Hamiltonian $H'$ is reflection positive
with respect to $\Theta$.

Applying Corollary \ref{cor:trafo}, we infer that the Gibbs state $\rho_{\beta H}$
for the \emph{original} Hamiltonian $H = \alpha^{-1}(H')$ is reflection positive for the 
gauge transformed 
reflection automorphism $\Theta' = \alpha^{-1}\Theta \alpha = \Theta \alpha^2$, given
in equation \eqref{eq:ferroreflection}.
\end{proof}

\subsection{Long-Range Interactions of Spin Pairs}\label{sec:LR}

The Heisenberg model with 
long-range interactions is defined by the Hamiltonian
\be\label{eq:lrheis2}
-H = \sum_{a = 1}^{3}  \sum_{\{x, x' \in \Lambda \, : \,x \neq x'\}} J^{a}\, 
\sigma^{a}_{x'}\sigma^{a}_{x} f(x - x')\,.
\ee
Here $f$ can be any reflection invariant, reflection positive function 
on $\R^{d}$, or on its compactification $\mathbb{T}^{m} \times \R^{d-m}$
in $m\leq d$ directions.
For such functions the matrix $f(\vartheta(x)-x')$ for $x,x'\in\Lambda_{+}$ is positive semidefinite.  Here there is extensive analysis, and some relevant papers are \cite{OS1,OSEuclideanFields,LuscherMack1975,GJ1979,Frank-Lieb2010}.

An important example is 
$f(x) = \|x\|^{-s}$ on $\R^{d}$, 
which is reflection positive
for $s \geq \max\{0, d-2\}$ by \cite[Proposition 6.1]{NeebOlafsson2014}.
Reflection positive functions on the compactification 
can be obtained from reflection positive functions on $\R^{d}$
under suitable conditions on the rapidity of their decay, see for example 
\cite[Proposition 15]{JaffeJaekelMartinez}. 

For long-range interactions, the matrix of coupling constants 
across the reflection plane will not be 
diagonal, as was the case for nearest neighbor models.

\subsubsection{Anti-Ferromagnetic Heisenberg Model}
For $J^{a}\leq 0$ (the anti-ferromagnetic case),
we can use 
the standard reflection $\Theta(\sigma^{a}_{j}) = - \sigma^{a}_{j}$.

\begin{prop}[Long-Range Heisenberg Model]\label{prop:LRHeis}
The Gibbs functional $\rho_{\beta H}$ 
for the Hamiltonian \eqref{eq:lrheis2}
is reflection positive 
with respect to $\Theta$
for all $\beta \geq 0$, 
if and only if $J^{a}\leq 0$ for $a = 1,2,3$.
\end{prop}
\begin{proof}
The Hamiltonian \eqref{eq:lrheis2} is 
hermitian and $\Theta$-invariant,
so by Theorem \ref{thm:maincorollarySpinSyst},
it is reflection positive for all $\beta \geq 0$ if and only if 
the matrix 
$i^{k_{\I} + k_{\I'}} J^{0\, AA'}_{\, \vartheta(\I)\I'}$
is positive semidefinite.

The matrix of coupling constants across the reflection plane 
has entries
\[
J^{0\, AA'}_{\, \vartheta(\I)\I'} = J^{a} f(\vartheta(x) - x')
\]
for the indices 
$(\I,A) = \{(x,a)\}$ and $(\I',A') = \{(x',a)\}$ 
of equation \eqref{eq:labeldef}, and
all other entries are zero. 
Since $k_{\I} = k_{\I'} = 1$, one finds
\[
i^{k_{\I} + k_{\I'}} J^{0\, AA'}_{\, \vartheta(\I)\I'}
=
-J^{a} f(\vartheta(x) - x')\,.
\]
As $f$ is reflection positive, this matrix 
is positive semidefinite if and only if  
$J^{a} \leq 0$ for $a = 1,2,3$.
\end{proof}

\subsubsection{Ferromagnetic Rotator Model}

The long-range rotator model is given by the Hamiltonian \eqref{eq:lrheis2}
with $J^{3} = 0$.

In the anti-ferromagnetic case $J^{1,2} \leq 0$, 
Proposition \ref{prop:LRHeis} shows that
it is reflection positive 
with respect to the standard reflection $\Theta$, satisfying 
$\Theta(\sigma^{a}_{j}) = - \sigma^{a}_{\vartheta(j)}$
for $a = 1,2,3$.
As in the nearest neighbor case, 
the ferromagnetic model $0 \leq J^{1,2}$ is 
reflection positive for a \emph{different} reflection $\Theta'$, satisfying
\eqref{eq:ferroreflection}.

\begin{prop}[\bf Long-Range Quantum Rotator]
The Gibbs state $\rho_{\beta H}$ for the Hamiltonian
\eqref{eq:lrheis2} is reflection 
positive with respect to the anti-linear reflection
$\Theta'$ for all $\beta \geq 0$,
if and only if $0 \leq J^{a}$ 
for $a = 1,2$.
\end{prop}

\begin{proof}
By Corollary \ref{cor:trafo}, $\rho_{\beta H}$ is reflection positive for 
$\Theta' = \alpha^{-1}\Theta\alpha$, 
if and only if $\rho_{\beta H'}$ is reflection positive for $\Theta$. 
Here  $H'= \alpha(H)$, and we choose
$\alpha = \alpha_{g}$ to be the gauge transformation
of equations \eqref{eq:gaugeno1} and \eqref{eq:gaugeno2}.

The gauge transformed Hamiltonian $H'$ has the form $H'= H'_{+} + H'_{0} + H'_{-}$,
where the term $H'_{0}$ containing the couplings across the reflection plane
is
\[
- H'_{0} = 
\sum_{x, x'\in \Lambda_{+}}  -\left(
J^{1} f(\vartheta(x) - x')\sigma^{2}_{\vartheta(x)}\sigma^{2}_{x'}
 + J^{2} 
f(\vartheta(x) - x')\sigma^{1}_{\vartheta(x)}\sigma^{1}_{x'}\right)\,.
\]
It follows that the matrix 
of couplings across the reflection plane 
for the gauge transformed Hamiltonian $H'$ is
\[
J'{}^{0\, AA'}_{\, \vartheta(\I)\I'} = -J^{\widehat{a}} f(\vartheta(x) - x')\,,
\]
for $(\I,A) = \{(x,a)\}$ and $(\I,A) = \{(x',a)\}$. Here $\widehat{a} = 1$ 
if $a = 2$ and vice versa.
Since $k_{\I} = k_{\I'} = 1$, the matrix 
$i^{k_{\I} + k_{\I'}} J'{}^{0\, AA'}_{\, \vartheta(\I)\I'}$ is positive semidefinite 
in the ferromagnetic case $0 \leq J^{a}$.

In order to apply Theorem \ref{thm:maincorollarySpinSyst} to $H'$,
we still need to check that $H'$ is reflection invariant under 
the standard reflection $\Theta$. 
By Corollary \ref{cor:trafo},
this is equivalent to reflection invariance of the original Hamiltonian 
$H$ under $\Theta'$.
This is readily seen to be the case 
by using the explicit equation \eqref{eq:ferroreflection} for $\Theta'$.

As $\alpha_{g}$ is a $*$-automorphism, 
$H'$ is hermitian, so $Z_{\beta H'} \geq 0$.
One then infers from Theorem \ref{thm:maincorollarySpinSyst}, 
that $\rho_{\beta H'}$ is reflection positive
with respect to $\Theta$.
As mentioned in the start of the proof, 
Corollary \ref{cor:trafo} then yields that 
$\rho_{\beta H}$ is reflection positive for $\Theta'$.
\end{proof}

\begin{Remark} \rm
An external field $\sum_{a=1}^{3}\sum_{j}h^{a}_{j}\sigma^{a}_{j}$
can be added to \eqref{eq:lrheis2} 
under the same conditions as in the nearest neighbor case.
For the anti-ferromagnetic Heisenberg model, 
$h^{a}_{\vartheta(j)} = - h^{a}_{j}$ for $a = 1,2,3$.
For the ferromagnetic quantum rotator, 
$h^{a}_{\vartheta(j)} =  h^{a}_{j}$ for $a = 1,2$, and 
$h^{a}_{\vartheta(j)} =  -h^{a}_{j}$ for $a=3$.
\end{Remark}

\begin{center}
{\bf Acknowledgements}\\[.4 cm]
\end{center}
\noindent The authors thank the MFO (Oberwolfach)
for hospitality 
at a conference where they first met, 
and had conversations leading to this research.  A.J.\ was supported in part by a grant from the Templeton Religion Trust.
B.J.\ was supported by the NWO grant 
613.001.214 ``Generalised Lie algebra sheaves."
He thanks A.J.\ for hospitality at Harvard University and in Basel during the writing of the paper.  

\bibliographystyle{alpha}
%\bibliography{List}	

\end{document}